\documentclass[sigconf]{acmart}
\renewcommand\footnotetextcopyrightpermission[1]{} 
\usepackage{amsmath,amssymb,amsfonts,amsthm}
\usepackage{epsfig,xspace,url,xcolor,color}
\usepackage[linesnumbered, ruled, algo2e, noresetcount]{algorithm2e}
\usepackage{setspace} 
\usepackage{graphicx}
\usepackage{textcomp}
\usepackage[caption=false]{subfig}
\usepackage[latin1]{inputenc}
\usepackage{float}
\usepackage{placeins}
\usepackage{listings}
\usepackage[T1]{fontenc}
\usepackage{caption}
\usepackage{enumitem}
\usepackage{csquotes}
\usepackage{hyperref}

\makeatletter  
\newcommand{\my@arrow}[1]{\ooalign{$#1-\mkern-8mu-$\cr\hidewidth$#1>$}}
\newcommand{\myarrow}{\mathrel{\mathpalette\my@arrow\relax}}
\makeatother

\newcommand{\vn}[1]{\textsf{\small #1}}
\newcommand{\fn}[1]{\textsc{#1}}

\newcommand{\junchang}[1]{\textcolor{blue}{[Junchang's note]}}

\newcommand{\sn}{\fn{DHash}\xspace}
\newcommand{\sns}{\fn{DHash}} 

\begin{document}
\sloppy

\title{\fn{\sns}: Enabling Dynamic and Efficient Hash Tables}

\author{Junchang Wang$^\star$ \ \ \ \ \ \ \ \ \ \ \ \ \ \ \ \ 
Xiong Fu$^\star$  \ \ \ \ \ \ \ \ \ \ \ \ \ \ \ \ 
Fu Xiao$^\star$ \ \ \ \ \ \ \ \ \ \ \ \ \ \ \ \ 
Chen Tian$^\dagger$}

\affiliation{%
  \institution{$^\star$Nanjing University of Posts and Telecommunications \ \ \ \ \ \ \ \ \ \ \ \ \ \ \ \ 
  $^\dagger$Nanjing University
  }
}

\begin{abstract}
  
Given a specified average load factor,
hash tables offer the appeal of constant time lookup operations.
However, hash tables could face severe hash collisions because of malicious attacks, 
buggy applications, or even bursts of incoming data,
compromising this practical advantage.
In this paper, we present \sns, a hash table that overcomes this challenge
by allowing programmers to dynamically change its hash function on the fly,
without affecting other concurrent operations such as
lookup, insert, and delete.
\sn is modular and allows programmers to select a variety of lock-free/wait-free set algorithms
as the implementation of hash table buckets.
With this flexibility, they can make trade-offs between the algorithm's progress guarantee,
performance, and engineering efforts, 
and create \sn implementations that meet their requirements best.
Evaluations on three types of architectures show that 
\sn noticeably outperforms other practical alternatives under heavy workloads.
With a load factor of 20, \sn outperforms the other three most widely used hash tables
by factors of 1.4-2.0,
and when the load factor increases to 200,
\sn is 2.3-6.2 times faster.
\end{abstract}



\maketitle

\section{Introduction}\label{sec.intr}

Given a specified average load factor,
hash tables bring the appeal to offer constant time lookup operations,
and hence have been widely used in operating systems and applications.
However, it is widely known
that hash tables are vulnerable to hash collisions \cite{Crosby2003DenialOS},
and randomizing static hash function is not a complete solution
\cite{dos2ht}.
For years, this vulnerability affected a long list of operating systems and programming languages
such as the Linux kernel \cite{linuxHT}, 
Perl \cite{Crosby2003DenialOS} and PHP \cite{phpdos}.

One possible solution is dynamically changing a hash table's hash function,
without affecting concurrent \emph{insert}, \emph{delete}, and \emph{lookup} operations
(henceforth simply \emph{common} operations).
We use the term \emph{dynamic} to describe a hash table
that can provide this flexibility feature,
and use the term \emph{rebuild} to describe
the function that dynamically changes its hash function.
Researchers have proposed several dynamic hash tables.
Herbert Xu created a dynamic hash table for 
the networking subsystem of the Linux kernel in 2010 \cite{Herbert2010Hashing}, 
to handle unpredictable large bursts of fragmented packets \cite{burst}
and potential DoS attacks.
Thomas Graf created another generic dynamic hash table in 2014 \cite{rhashtable},
which since then has been widely used in the Linux kernel.
Other researchers partially overcame this challenge
by proposing hash table algorithms that can only
enlarge or shrink bucket sizes by a factor of 2
\cite{Shalev2006SplitorderedLL,Triplett2011ResizableSC,Fatourou2018AnEW},
which we refer to as \emph{resizable} hash tables.

The core problem in designing a dynamic hash table is
how to atomically distribute each node from the old hash table
to the new one in rebuilding.
Prior research activities overcame this challenge by using various techniques
(detailed in Section \ref{sec.related}),
but in practice, we found that they have performance drawbacks
when used in scenarios with heavy workloads, bursts of incoming data,
and/or attacks.
For example, Xu's and Graf's algorithm uses per-bucket locks
to serialize concurrent update operations,
which leads to severe contentions
when the load factors increase to 20 and more.
Resizable hash tables do not need to face this challenge,
but they have limited capabilities to solve hash collisions.

This paper presents \sn (Dynamic Hash table),
a dynamic hash table that
can meet the following main goals.

\noindent
\textbf{(1) Dynamic hash table}:
Users can dynamically change the hash function,
without affecting concurrent common hash table operations.

\textit{Rationale}:
Dynamic hash table is the algorithm of choice for critical applications
facing bursts of update requests, buggy applications, and even malicious attacks.

\noindent
\textbf{(2) Modularity}:
The hash table should be modular; 
it can utilize various lock-free/wait-free set algorithms
as the implementation of hash table buckets, 
without heavy engineering workload.

\textit{Rationale}:
The choice of the algorithm to solve conflicts within each bucket 
is a trade-off between the algorithm's progress guarantee, performance, and engineering efforts.
In practice, many users of hash tables cannot know the
right choice of the algorithm in advance.
For example, wait-free linked lists \cite{kogan2011wait, yang2016wait} are the algorithms of choice 
for users who want the strongest progress guarantee for common hash table operations,
and users who look for fast lookup speed
would choose lock-free skip lists \cite{fomitchev2004lock}.

\noindent
\textbf{(3) Fast and non-blocking lookup operations}:
Lookup operations should be fast and non-blocking,
no matter if a rebuild operation is in progress.

\textit{Rationale}:
Hash tables are commonly designed for use cases
with significant reads than writes.
For example, Herlihy and Shavit suggested a common workload for hash tables
with 90\% lookups and 10\% insertions along with deletions \cite{TheArt}.

\noindent
\textbf{(4) Fast and non-blocking update operations}:
Insert and delete operations should be fast and non-blocking,
even when a rebuild operation is distributing nodes.

\textit{Rationale}:
In modern computer systems, insert and delete requests typically
reach hash tables in batch.
For example, a variety of places in between two servers
(e.g., buffers in hardware Network Card and kernel TCP/IP stacks) 
can buffer network packets and then send them out in batch for higher throughput \cite{Alizadeh2012LessIM}.
Failing to handle the large bursts of update requests could result in 
performance degradation \cite{burst}.

The core of \sns's technical contribution is an efficient rebuilding strategy
that can distribute nodes by using regular operations.
The rebuilding strategy allows \sn to leverage a variety of lock-free/wait-free set algorithms as hash buckets,
without heavy engineering workloads. 
Experimental data shows that 
with light workloads, the overall performance of \sn 
matches or slightly exceeds other practical representatives including
the two dynamic hash table in the Linux kernel \cite{Herbert2010Hashing, linuxHT}
and one resizable hash table based on split-ordered-list \cite{Shalev2006SplitorderedLL}.
\sn noticeably outperforms these algorithms by factors of 2.3-6.2 and more under heavy workloads.

The rest of the paper is organized as follows.
We first discuss related work in Section \ref{sec.related}.
Section \ref{sec.arch.overview} gives an overview of \sns.
Section \ref{sec.arch.impl} presents the details of the algorithm.
We prove the correctness of \sn in Section \ref{sec.correcness},
present evaluations in Section \ref{sec.eva},
and conclude in Section \ref{sec.conclusion}


\section{Related Work}\label{sec.related}

This section sketches a high-level overview of
prior researches on dynamic and resizable hash tables.

\textbf{Herbert Xu's dynamic hash table}\label{sec.related.herbert}:
Herbert Xu created a dynamic hash table \cite{Herbert2010Hashing}
for the management of IGMP packets in the Linux kernel in 2010.
As far as we know, this is the first practical dynamic hash table.
The key idea behind Xu's algorithm is to
manage two sets of pointers in each node,
so that common operations
traverse one set of pointers while the rebuild operation
is updating the other set.
The two sets are exchanged upon the completion of every rebuild operation.
One major benefit of introducing two sets of pointers is that
it is not necessary to delete nodes from the old hash table
while a rebuild operation is in progress.

Xu's hash table algorithm is straightforward and easy to be implemented,
but it has two major drawbacks in practice.
(1) Each bucket contains a common linked list
along with a lock to serialize concurrent update and rebuild operations to the bucket.
(2) A linked list algorithm
must be customized by adding an extra set of pointers,
before it can be used by Xu's hash table.
This not only results in increased memory footprint,
but also prevents Xu's algorithm from utilizing other 
faster linked list algorithms.
In contrast, \sn overcomes these drawbacks in its design.


\textbf{Generic dynamic hash table in Linux kernel}\label{sec.related.rhashtable}:
Thomas Graf introduced a generic dynamic hash table into the Linux kernel 
in 2014 \cite{rhashtable}, and this algorithm has been widely used in the kernel.
Graf's algorithm was originally based on Josh Triplett's ATC'11 paper \cite{Triplett2011ResizableSC},
but has significantly improved in the performance of its rebuild operation.

Graf's hash table maintains a single pointer in each node,
and utilizes a mutex lock to synchronize concurrent update and rebuild operations
to the same bucket.
A rebuild operation traverses the hash table, and always finds a non-empty bucket 
and distributes its last node, by inserting this node into the new hash table 
and then deleting it from the old hash table.
There is a time period during which 
the node can be found in both the old and the new hash table,
and a lookup operation searching the old hash table could be 
erroneously redirected to the new hash table. 
So the lookup operation of Graf's algorithm is designed to tolerate these behaviors.

Graf's algorithm is a practical design. 
However, this algorithm has the following drawbacks.
(1) The rebuild thread must reach the tail of a list to distribute a single node.
(2) It uses locks to serialize updates to a single bucket.
(3) It maintains unordered lists as its buckets, 
which noticeably increases the overhead of lookup operations.
In contrast, \sn overcomes these drawbacks in its design.

\textbf{Resizable hash tables}
The resizable hash tables do not change their hash functions;
they can only enlarge and shrink their bucket sizes by a factor of 2.
Ori Shalev presented the first lock-free resizable hash table
\cite{Shalev2006SplitorderedLL} in 2006.
This algorithm keeps a singly linked list,
and all nodes in the hash table are chained in this list.
In solving the atomic-distribution problem in resizing,
Shalev's algorithm does not move nodes among the buckets,
instead, it moves the buckets among the nodes
by referencing buckets to the proper nodes in the list.
Shalev's algorithm, even though is lock-free, has drawbacks in practice.
(1) It must use a modulo $2^i$ hash function,
which dramatically limits the flexible of the algorithm.
(2) The algorithm must first reverse the bit string of the key
before performing any operations.
Unfortunately, the reverse operation is not always efficient
on platforms where hardware cannot provide special instructions
for bit string processing.

Josh Triplett presented a hash table that can incrementally
shrink and expand by chaining multiple lists together and by 
splitting existing lists, respectively \cite{Triplett2011ResizableSC}.
Triplett's algorithm has drawbacks. For example,
buckets are implemented as unordered lists,
and concurrent insert and delete operations must block until
a concurrent shrink operation finishes publishing the new hash table.

Researchers have proposed wait-free resizable hash tables
that provide the strongest progress guarantee
\cite{liu2014dynamic, Fatourou2018AnEW}.
However, it is not clear how features such as 
duplicated nodes and node replacement,
which is commonly desired in practice,
can be implemented in these wait-free resizable hash tables.

\section{\sn Algorithm Overview}\label{sec.arch.overview}

This section first presents the challenge in designing dynamic hash tables,
and then sketches a high-level overview of the rebuild, lookup, insert, and
delete operations of \sns,
leaving technical details to Section \ref{sec.arch.impl}.

\begin{figure*}[t]
    \centering
    \subfloat[Initial state. The hash table contains two buckets.]{ 
        \includegraphics[scale=0.23, angle=0]{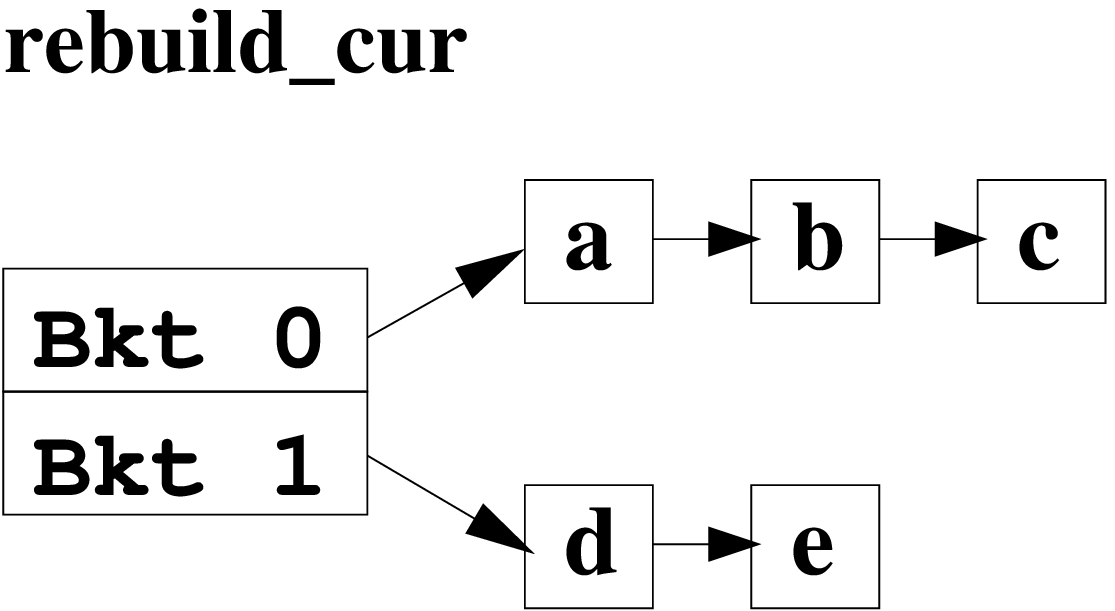}
        \label{fig.rebuild.1}
    }
    \hspace{0.1cm}
    \subfloat[A rebuild operation starts. 
	        \emph{rebuild\_cur} first points to node \textbf{a}.]{
        \includegraphics[scale=0.23, angle=0]{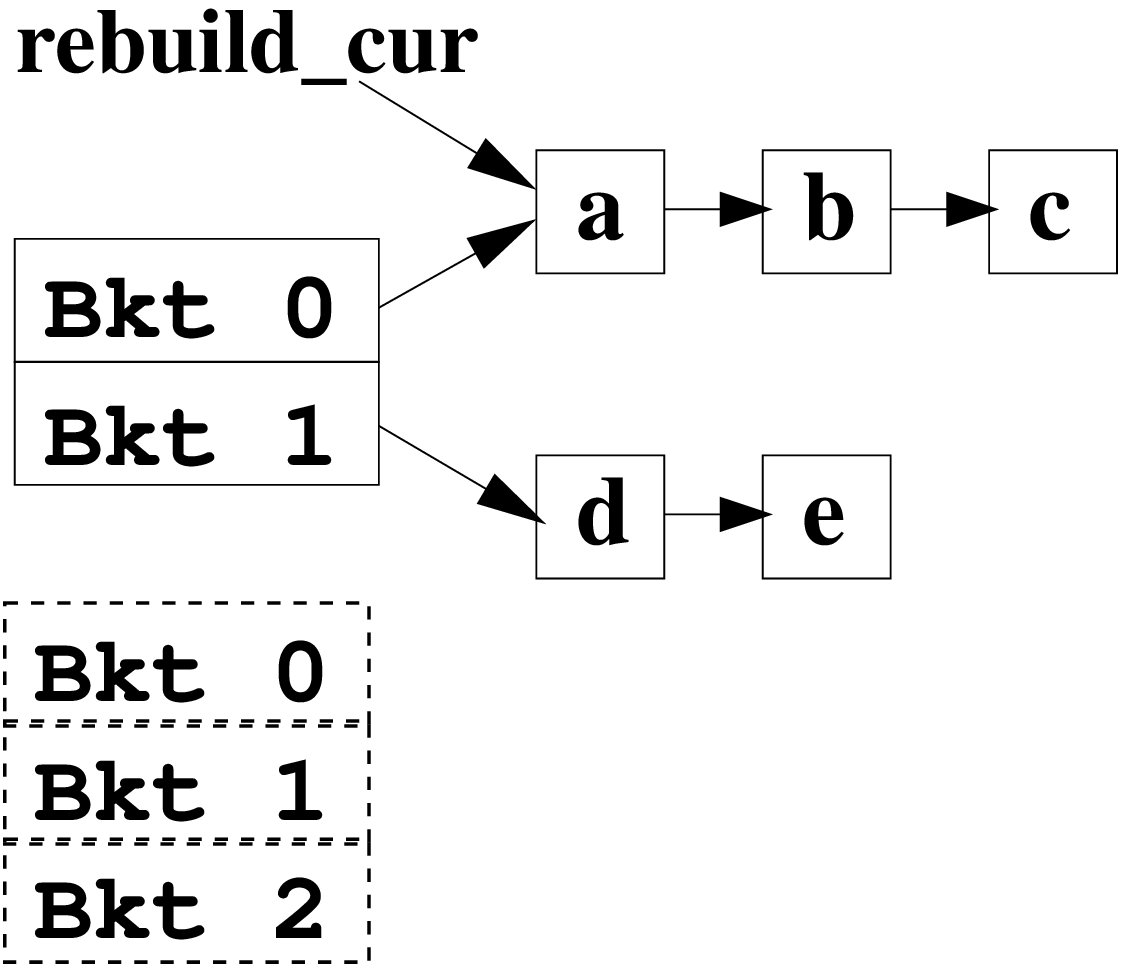}
        \label{fig.rebuild.2}
    }
    \hspace{0.1cm}
    \subfloat[\textbf{a} is deleted from the old hash table.
            \textbf{a} is in \emph{hazard period}.]{
        \includegraphics[scale=0.23, angle=0]{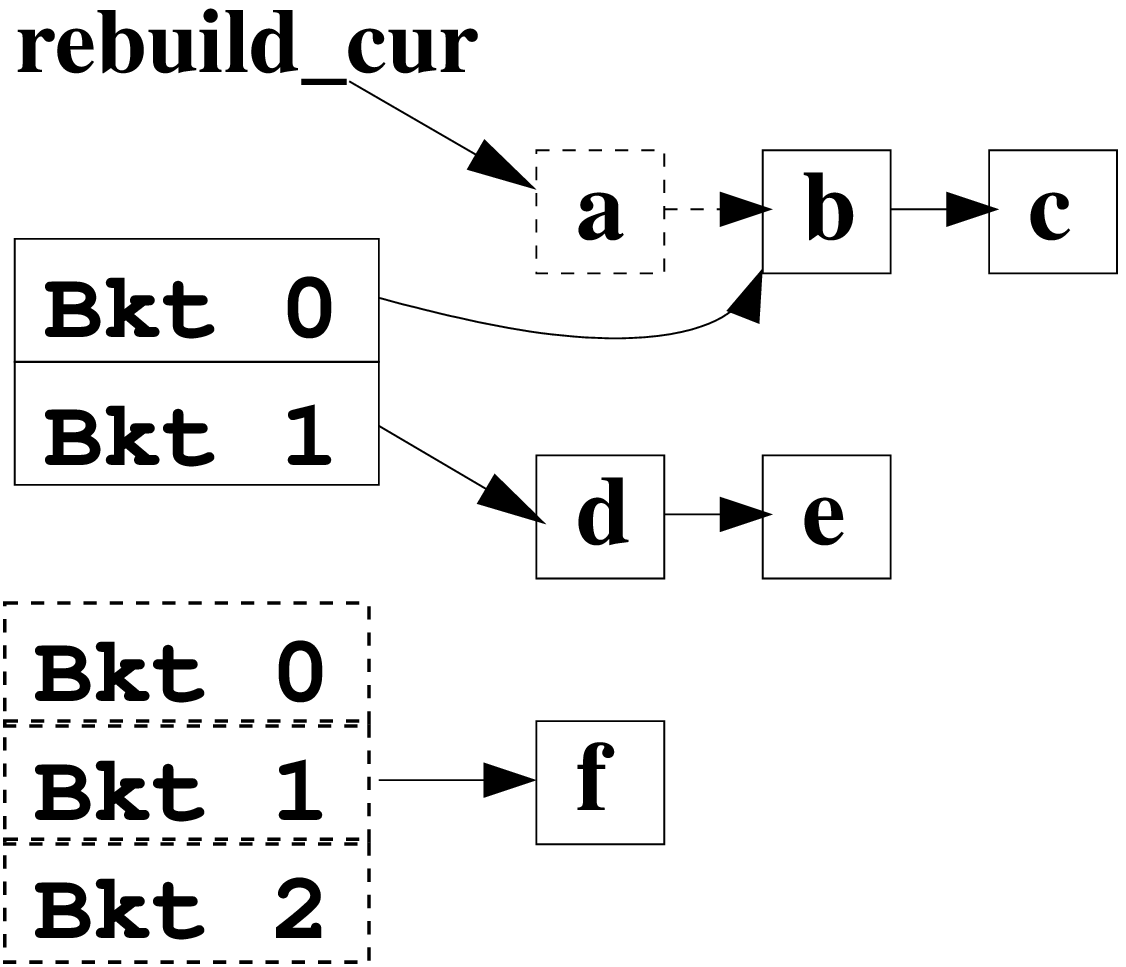}
        \label{fig.rebuild.3}
    }
    \hspace{0.1cm}
    \subfloat[\textbf{a} is inserted into the new table.
		    \emph{rebuild\_cur} is then set to \emph{NULL}.]{
        \includegraphics[scale=0.23, angle=0]{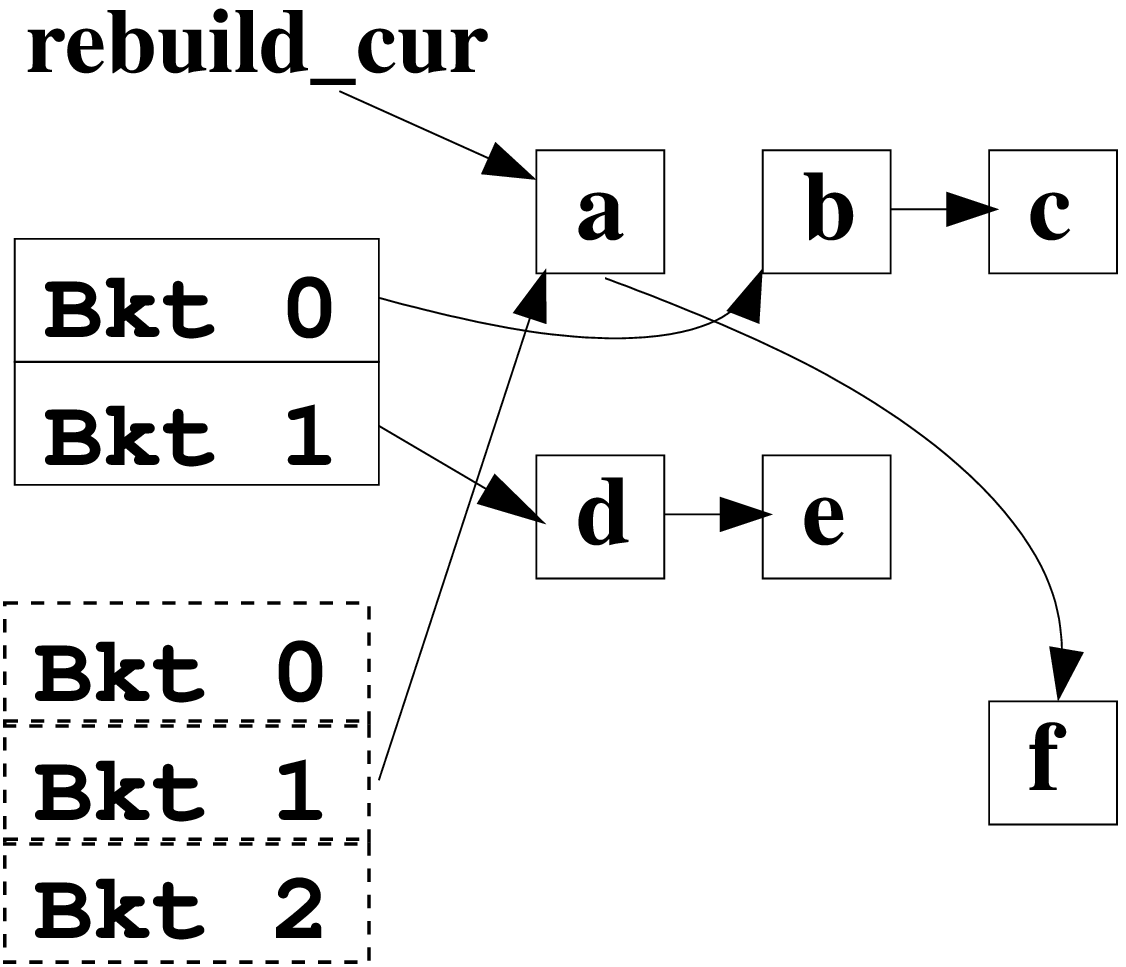}
        \label{fig.rebuild.4}
    }
    \hspace{0.1cm}
    \subfloat[Rebuild traverses the old hash table, and distributes all of the nodes.]{
    \includegraphics[scale=0.23, angle=0]{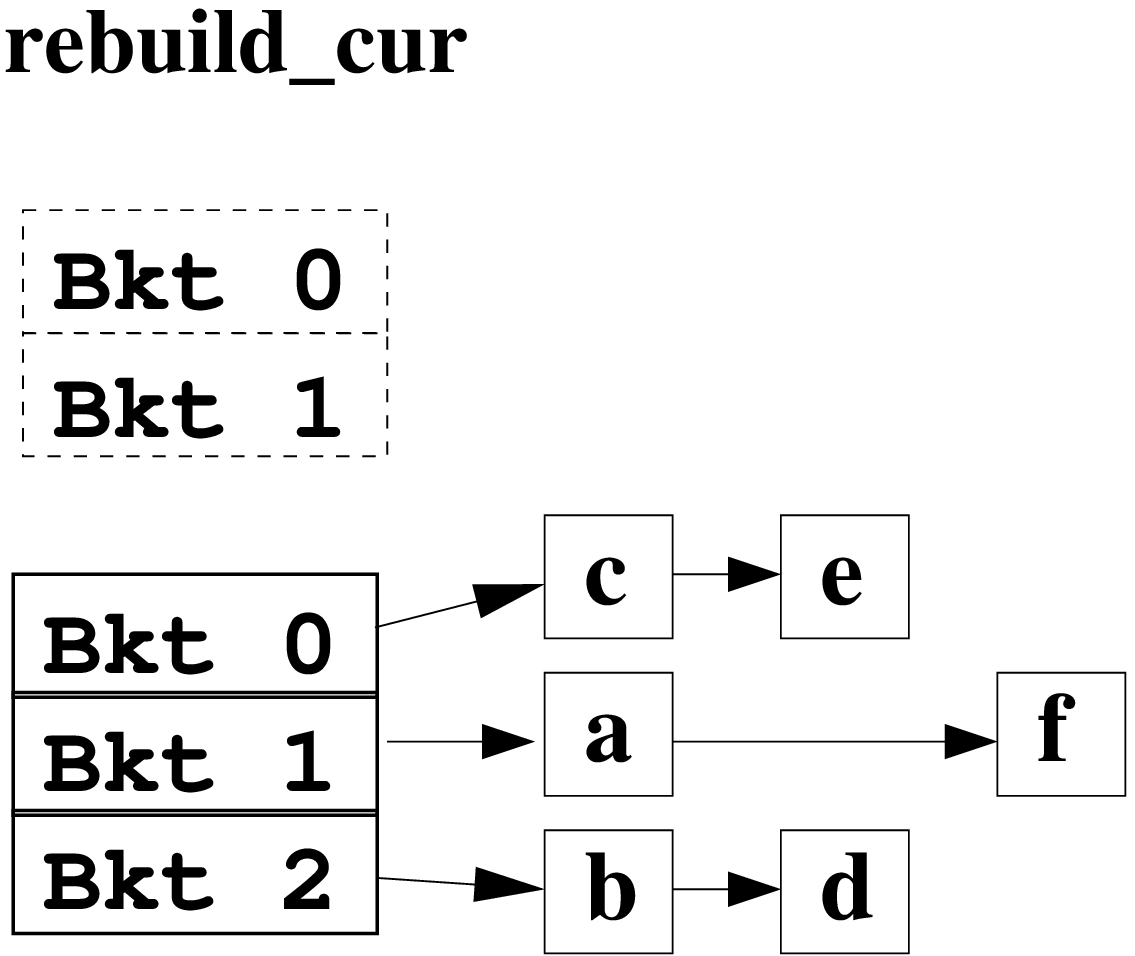}
    \label{fig.rebuild.6}
    }
    \hspace{0.1cm}
    \subfloat[Rebuild waits for prior operations, then frees the old hash table.]{
    \includegraphics[scale=0.23, angle=0]{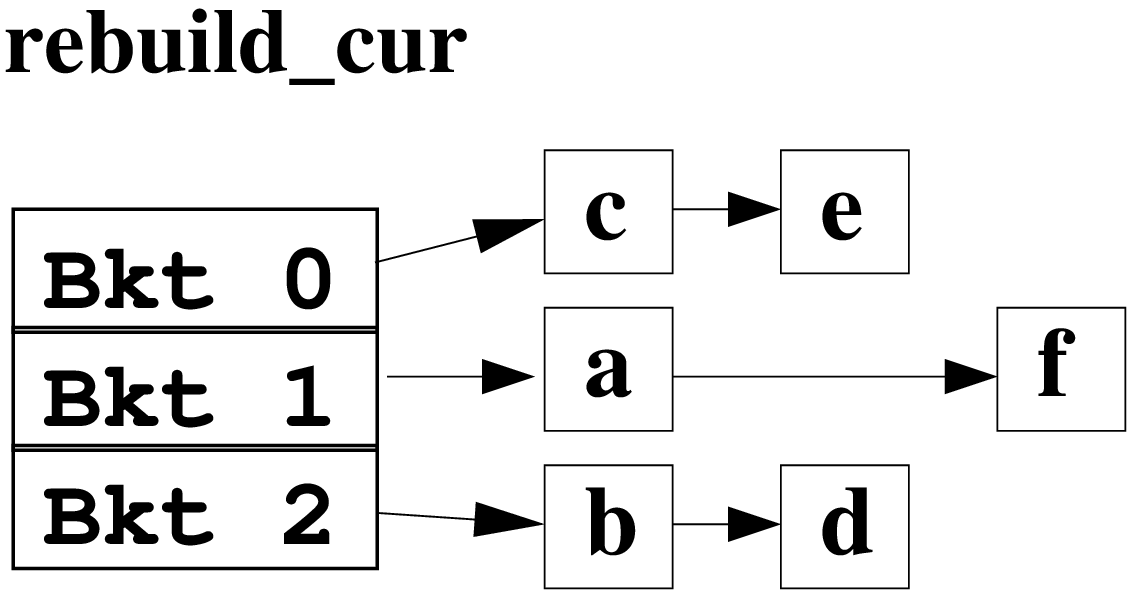}
    \label{fig.rebuild.7}
    }
    \caption{The flow of a \emph{rebuild} operation.
        \emph{Bkt} is short for \emph{Bucket}.
        \emph{rebuild\_cur} is referencing to the node in \emph{hazard period}.}
    \label{fig.rebuild}
\end{figure*}

\textbf{Challenges}
The key challenge in designing a \emph{dynamic} hash table
is to atomically move nodes from the old hash table to the new one,
without affecting concurrent common hash table operations.



This challenge is hard to be efficiently handled
because moving a node must touch more than two buckets (linked lists),
and the transition of the node must be atomic with respect to other concurrent operations.
Prior work solved this issue by 
(1) introducing two sets of linked lists for each bucket \cite{Herbert2010Hashing},
(2) acquiring corresponding per-bucket locks before distributing each node \cite{Herbert2010Hashing, rhashtable},
(3) avoiding moving nodes by adjusting bucket pointers \cite{Shalev2006SplitorderedLL},
and/or (4) maintaining unordered linked lists that sometimes may contain nodes
that do not belong to the linked lists \cite{Triplett2011ResizableSC}.
These approaches, however, sacrifice the algorithms' generality/performance.
\sns, inspired by the RCU technique, takes a fundamentally different approach
by releasing the atomicity requirement in distributing a node;
the rebuild operation first deletes the node from the old hash table
and then inserts it into the new hash table,
by using regular operations,
rather than expensive synchronization and memory fence instructions.
The process of distributing the node leads to a short time period
during which in neither hash tables can this node be found.
We call this a node's \emph{hazard period}.
To allow other operations to be able to access the node,
\sn employs a global pointer that always points to the node that is in \emph{hazard period}.
On the other hand, if a rebuild thread is in progress, 
other operations need to check different locations
because a node may reside in either the new or the old hash table,
or is being referenced by the global pointer.
In Section \ref{sec.arch.impl},
we prove that if a lookup, insert, or delete operation
checks both hash tables and checks the node currently in \emph{hazard period} 
in a specified order, they can always find the node with the matching key and perform correctly.


\subsection{Rebuild operation}\label{sec.arch.overview.rebuild}


\sn consists of a specified number of buckets.
To solve collisions between multiple keys that hash to the same bucket,
a lock-free linked list
is used to chain together nodes containing these keys.
The resulting data structure of \sn is illustrated in Figure \ref{fig.rebuild.1}.
In this example, \sn consists of two buckets: \emph{Bkt 0} and \emph{Bkt 1}.
\emph{Bkt 0} contains three nodes (\emph{a}, \emph{b}, and \emph{c}),
and \emph{Bkt 1} contains two nodes (\emph{d} and \emph{e}).

When a rebuild operation starts,
we assume that the new hash table contains three buckets,
and that the users provide a new hash function that can map
all of the keys to the new three-bucket array.
The rebuild operation performs a hash table traversal and
distributes nodes in the old hash table to the new one.
Figure \ref{fig.rebuild} illustrates the process of moving node \vn{a} to the new hash table,
with the initial state of two buckets shown in Figure \ref{fig.rebuild.1}
and with the time advancing from figure to figure.

Specifically, node \vn{a}
is first referenced by the global pointer \emph{rebuild\_cur}, 
resulting in the state shown in Figure \ref{fig.rebuild.2}.
Then, node \vn{a} is removed from the old hash table and enters its \emph{hazard period},
shown in Figure \ref{fig.rebuild.3}.
When node \vn{a} is in \emph{hazard period},
other lookup and update operations can access it
 via the global pointer \emph{rebuild\_cur}.
Allowing other threads to access the node that is in \emph{hazard period}
is the key reason that distributing a node is not necessary to involve 
expensive atomic operations and memory fences in \sns.
Without loss of generality, we assume that when the rebuild operation is in progress,
other operations concurrently insert a new node, \vn{f}, into the new hash table,
shown in Figure \ref{fig.rebuild.3}.

Then, node \vn{a} is inserted into the new hash table, shown in Figure \ref{fig.rebuild.4}.
After it has been successfully inserted into the new hash table,
\emph{rebuild\_cur} is set to NULL.
The rebuild operation traverses the old hash table
and distributes all of the nodes to the new hash table.
Then the rebuild operation exposes the new hash table to subsequent 
common operations,
shown in Figure \ref{fig.rebuild.6}.
After that, the rebuild operation 
waits for all prior unfinished operations to complete,
before safely reclaiming the old hash table,
shown in Figure \ref{fig.rebuild.7}.


\subsection{Lookup, Insert, Delete operations}\label{sec.arch.overview.others}

When rebuild operations are absent, which is the common case,
a lookup, insert, and delete operation performs common hash table operations
on the only hash table in \sns.
When a rebuild operation is in progress,
a node may be in \emph{hazard period}
during which in neither hash tables can this node be found.
Hence, a lookup, delete, and insert operation must comply with this variation.

\textbf{Lookup}:
A lookup operation must search both the hash tables
and check the node referenced to by the global pointer \emph{rebuild\_cur}.
Since the operation involves multiple shared memory spaces,
synchronizing the rebuild and the lookup operation
is a classic synchronization problem,
which we solve by managing the lookup operation to check these memory spaces
in a specified order.
Specifically, the lookup operation first searches for the node in the old hash table,
then checks if the node pointed to by \emph{rebuild\_cur} has the matching key,
and finally searches in the new hash table.
Lemma \ref{lemma.lookup} in Section \ref{sec.arch.impl.lookup} proves that
if a lookup operation performs in this order,
it can find the node with the matching key, 
no matter if a rebuild operation is in progress.


\textbf{Delete}:
Similarly,
a delete operation first searches in the old hash table
and deletes the node and returns if the node can be found.
Otherwise, it checks the node pointed to by \emph{rebuild\_cur}.
Finally, the delete operation searches in the new hash table.
Lemma \ref{lemma.deletion} in Section \ref{sec.arch.impl.delete} proves that
if a delete operation performs in this order,
it can successfully find and delete the node with the matching key.

\textbf{Insert}:
A rebuild operation waits until all prior unfinished operations have completed
before replacing the old hash table with the new one 
(detailed in Section \ref{sec.arch.impl.rebuild}).
Lemma \ref{lemma.insert} in Section \ref{sec.arch.impl.insert} proves that
when a rebuild operation is in progress,
the insert operation can simply
insert the node into the new hash table and then return.

\section{\sn implementation}\label{sec.arch.impl}

The design of \sn presented in Section \ref{sec.arch.overview} leads to a
relatively straightforward implementation, which is the subject of this
section.

\subsection{Preliminaries}\label{sec.bb.rcu}\label{sec.bb.list}

We first give a brief overview of Read-Copy Update (RCU)
which \sn uses to synchronize concurrent operations.
Note that \sn can also use other synchronization mechanisms
such as reference counters \cite{valois1995lock} 
and hazard pointers \cite{michael2004hazard}.
We then present an RCU-based lock-free linked list
which is used as the implementation of hash table buckets.
Note that \sn is modular, such that the linked list 
can be replaced by other lock-free/wait-free set algorithms.

\textbf{Read-Copy Update}:
RCU distinguishes between read-side code and write-side code
and has the following primitives to synchronize read-write conflicts:

\begin{itemize}[leftmargin=*]
\item \vn{rcu\_read\_lock() / rcu\_read\_unlock()}:
Each time a thread wants to access shared variables,
it accesses them in a read-side critical section,
which begins with the primitive \vn{rcu\_read\_lock()} 
and ends with the primitive \vn{rcu\_read\_unlock()}.
Within a read-side critical section,
the lookup thread is safe to access the shared resources
without needing to worry about the potential issues
that these resources could be freed by other threads at the same time. 
\item \vn{synchronize\_rcu()}: works as a
wait-for-readers barrier.
Each time an updater thread wants to update shared variables 
(e.g., to delete a node),
it first makes the resources unreachable to subsequent lookup operations,
and then invokes \vn{synchronize\_rcu()} to wait
until existing lookup operations to safely complete
before updating the shared variables.
\item \vn{call\_rcu()} is an asynchronous version of \vn{synchronize\_rcu()}.
It can be used by updater threads that do not want to block.
\end{itemize}

RCU synchronizes readers with writers by using constrained access order,
instead of shared variables.
Any RCU-protected node accessed during
a reader is guaranteed to remain unreclaimed
until the reader completes its access
and calls \emph{rcu\_read\_unlock()}.
The production-quality implementations of \emph{rcu\_read\_lock()}
and \emph{rcu\_read\_unlock()} are extremely lightweight;
they have exactly zero overhead in Linux kernels
built for production use with \emph{CONFIG\_PREEMPT=n}
\cite{rcu}
and have extremely close to zero overhead in user-space applications
when the \emph{QSBR flavor} model is used \cite{userspace-rcu},
such that readers of RCU-based data structures can perform
as fast as single-threaded programs.


\textbf{RCU-based lock-free linked list}:
For ease of presentation,
in this paper
we chose Michael's lock-free linked list \cite{Michael2002Hashing}
as the implementation \sns's buckets.
We optimized Michael's algorithm to meet our design goals better.
Specifically,
Michael's original algorithm uses \emph{hazard pointers} \cite{michael2004hazard} 
to synchronize concurrent access to shared variables,
which is robust but involves expensive programming and run-time overhead.
A 64-bits long \emph{tag} field must be added for each node
to prevent the potential ABA-problem \cite{TheArt}.


To overcome these problems,
we created an RCU-based lock-free linked list,
which is based on Michael's algorithm
but leverages \emph{RCU} to efficiently manage read-write conflicts.
The modifications are as follows.
(1) The RCU technique instead of \emph{hazard pointers} is used
as the memory reclamation scheme,
such that the expensive memory fences in traversing the list can be removed.
(2) The \emph{tag} field in each node is saved,
because the RCU technique prevents reclaiming (and hence reusing) nodes
before concurrent lookup operations holding references to these nodes have completed.
(3) To reclaim a node, \emph{call\_rcu} is used, 
such that a delete operation will not be blocked by prior unfinished lookup operations.


\begin{algorithm2e}[ht]
    \footnotesize
    \SetStartEndCondition{ (}{)}{)}\SetAlgoBlockMarkers{}{\}}%
    \SetKwProg{Fn}{}{ \{}{}\SetKwFunction{FRecurs}{void FnRecursive}%
    \SetKwFor{For}{for}{ \{}{}%
    \SetKwIF{If}{ElseIf}{Else}{if}{ \{}{elif}{else \{}{}%
    \SetKwFor{While}{while}{ \{}{}%
    \SetKwRepeat{Repeat}{repeat\{}{until}%
    \AlgoDisplayBlockMarkers\SetAlgoNoLine%
    
    \SetKwInOut{Input}{Input}
    \SetKwInOut{Output}{Output}
    \SetKwInOut{Parameters}{Parameters}
    \SetKwInOut{Variables}{Variables}
    
    \textbf{struct} node \{long key; <node *, flag> next\}\;
    \textbf{define} LOGICALLY\_REMOVED         (1UL < < 0) \\
    \textbf{define} IS\_BEING\_DISTRIBUTED        (1UL < < 1) \\
    \textbf{struct} lflist \{node *head\}\;
    \textbf{struct} snapshot \{node **prev, *cur, *next\}\;
    
    \BlankLine
    \tcc{Search the node with the matching \emph{key} in list \emph{htbp}.}
    \textit{lflist\_find(lflist *htbp, key, snapshort *sp)}

    \tcc{Insert node \emph{htnp} into list \emph{htbp}.}
    \textit{lflist\_insert(lflist *htbp, node *htnp)}

    \tcc{Search the node with the matching \emph{key} in list \emph{htbp},
         set the node's \emph{flag} bit, and try to physically delete the node.}
    \textit{lflist\_delete(lflist *htbp, long key, long flag)}
    
    \caption{Structures and APIs of our lock-free linked list.}
    \label{alg.structure.list}
\end{algorithm2e}

Data structures and the API set of our RCU-based lock-free linked list is presented 
in Algorithm \ref{alg.structure.list}.
The structure \emph{lflist},
which will be used as the implementation of \sns's hash table buckets,
is fundamentally a chain of \emph{nodes}.
For each \emph{node},
the \emph{key} field holds the key value,
the \emph{next} field points to the following node in the linked list if any,
or has a \emph{NULL} value otherwise.
Since pointers are at least word aligned on all currently available architectures,
the two least significant bits of \emph{next}
are used as the \emph{flag} field
indicating if the node is in a special state.
The least significant bit, denoted as \emph{LOGICALLY\_REMOVED},
is used to indicate that a node has been logically removed by a delete operation.
The second to the least significant bit, denoted as \emph{IS\_BEING\_DISTRIBUTED},
is used to indicate that the node has been logically removed from the list
by a rebuild operation.
The difference between these two states is how the node will be reclaimed,
which we will discuss in detail in the following paragraphs.

Structure \emph{snapshot} is to return the search result
to the function invoking \emph{lflist\_find}.
Each time we want to search a node,
an instance of \emph{snapshot} is passed to \emph{lflist\_find}.
Upon the completion of \emph{lflist\_find},
it is guaranteed that the \emph{cur} field of the snapshot points to
the list node which contains the value
that is greater than or equal to the specified search key,
and that \emph{prev} and \emph{next} fields
point to its predecessor node and following node, respectively.

Our RCU-based lock-free linked list provides three basic operations,
\emph{lflist\_find}, \emph{lflist\_insert}, and \emph{lflist\_delete}, as shown in Algorithm \ref{alg.structure.list}.
Before invoking any of these functions, a caller must have entered the RCU read-side critical section
by invoking rcu\_read\_look().
Function \emph{lflist\_insert} and \emph{lflist\_delete} also need the read-side protection because
they need to first traverse the list.
Function \emph{lflist\_delete} takes the third parameter \emph{flag},
which is first stored to the \emph{flag} field of the target node.
Function \emph{lflist\_delete} deletes the matching node from the list and reclaims the node memory
if \emph{flag} is set to \emph{LOGICALLY\_REMOVED}.
In contrast, if \emph{flag} is set to \emph{IS\_BEING\_DISTRIBUTED},
the node memory will not be reclaimed
because the node will be inserted into the new hash table later.
Function \emph{lflist\_delete} does not block;
it uses \emph{call\_rcu} to asynchronously reclaim a node.
Note that \emph{call\_rcu} is safe to be invoked within an RCU read-side critical section.


\subsection{Data structures}\label{sec.arch.impl.ds}

\begin{algorithm2e}[ht]
    \footnotesize
    \SetStartEndCondition{ (}{)}{)}\SetAlgoBlockMarkers{}{\}}%
    \SetKwProg{Fn}{}{ \{}{}\SetKwFunction{FRecurs}{void FnRecursive}%
    \SetKwFor{For}{for}{ \{}{}%
    \SetKwIF{If}{ElseIf}{Else}{if}{ \{}{elif}{else \{}{}%
    \SetKwFor{While}{while}{ \{}{}%
    \SetKwRepeat{Repeat}{repeat\{}{until}%
    \AlgoDisplayBlockMarkers\SetAlgoNoLine%
    
    \SetKwInOut{Input}{Input}
    \SetKwInOut{Output}{Output}
    \SetKwInOut{Parameters}{Parameters}
    \SetKwInOut{Variables}{Variables}
    
    \textbf{struct} ht \{ht *ht\_new; long (*hash)(long key);
                         int nbuckets; lflist *bkts[]\}\;

    \tcc{Global variables}
    struct node *rebuild\_cur\;
    mutex rebuild\_lock\;

    \BlankLine
    \textit{clean\_flag(node *htnp, long flag)} \textbf{\{} atomic\_fetch\_and(htnp->next, $\sim$flag); \textbf{\}}\label{alg.structure.flag1} \\

    \textit{set\_flag(node *htnp, long flag)} \textbf{\{} atomic\_fetch\_or(htnp->next, flag); \textbf{\}}\label{alg.structure.flag2} \\

    \Fn{ht\_alloc(int nbuckets, long (*hash)(long key))}{
        htp->ht\_new := NULL;    htp->hash := hash\;
        htp->nbuckets := nbuckets;  htp->bkts := allocate(nbuckets)\;
        
    }
    
    \caption{Structures and helper functions of \sns.}
    \label{alg.structure}
\end{algorithm2e}

Algorithm \ref{alg.structure} lists the data structures and auxiliary functions of \sns.
The main structure, \emph{ht}, is an array of buckets (\emph{bkts[]}) of size \emph{B},
where \emph{B} is specified by the \emph{nbuckets} field.
Each element of \emph{bkts} is fundamentally a pointer to our RCU-based lock-free
linked list \emph{lflist}.
The \emph{hash} field is a function pointer to the user-specified hash function.
The \emph{ht\_new} field is set to \emph{NULL}
unless a rebuild operation is in progress,
in which case it points to the new hash table
that is going to replace the old one.
The global variable 
\emph{rebuild\_cur} points to the node that is currently in \emph{hazard period}
or equals to \emph{NULL} if there is no such a node in the system,
and the mutex lock \emph{rebuild\_lock} is to serialize attempts to rebuild the hash table.

The two helper functions, \emph{clean\_flag} and \emph{set\_flag},
cleans or sets the \emph{flag} bits 
of the node pointed to by \emph{htnp}.
Since the \emph{next} field of a node could be updated 
by concurrent operations, these two operations must perform atomically.
The helper function \emph{ht\_alloc} creates a hash table,
by allocating the array of buckets
and assigning the user-specified hash function to the \emph{hash} field.


\subsection{Solving read-write conflicts}\label{sec.alg.details.rwconflicts}

There are read-write conflicts between \sns's lookup and delete operations,
and between common operations and rebuild operations.
\sn solves this issue by leveraging the RCU synchronization mechanism,
which is discussed in Section \ref{sec.bb.rcu}.
Specifically, a caller must first enter an RCU read-side critical section
before invoking \sns's common operations and referencing any node in \sns,
shown in the following code snippet.

\begin{footnotesize}
\begin{verbatim}
    rcu_read_lock();
    node *cur = ht_lookup(htp, key);
    /* Accessing *cur is safe here. */
    rcu_read_unlock();
    /* Accessing *cur becomes unsafe. */
\end{verbatim}
\end{footnotesize}

\subsection{Rebuild operation}\label{sec.arch.impl.rebuild}

The rebuild operation is shown in Algorithm \ref{alg.rebuild}.
Line \ref{alg.rebuild.trylock} attempts to acquire the global lock \emph{rebuild\_lock},
which serializes concurrent rebuild requests.
Once \sn has the lock, it checks again that the rebuild is still required
on line \ref{alg.rebuild.doublecheck}.
Line \ref{alg.rebuild.alloc} allocates a new hash table
which has the user-specified size and hash function.
Line \ref{alg.rebuild.assign1} assigns the reference to the new hash table to
the \emph{ht\_new} field of the old hash table,
allowing subsequent operations to access the new hash table.
Line \ref{alg.rebuild.sync1} performs an RCU synchronization barrier
to wait for prior common operations, which may not be aware of the new hash table,
to complete before the rebuild operation continues.

\begin{algorithm2e}[ht]
    \footnotesize
    \SetStartEndCondition{ (}{)}{)}\SetAlgoBlockMarkers{}{\}}%
    \SetKwProg{Fn}{}{ \{}{}\SetKwFunction{FRecurs}{void FnRecursive}%
    \SetKwFor{For}{for}{ \{}{}%
    \SetKwIF{If}{ElseIf}{Else}{if}{ \{}{elif}{else \{}{}%
    \SetKwFor{While}{while}{ \{}{}%
    \SetKwRepeat{Repeat}{repeat\{}{until}%
    \AlgoDisplayBlockMarkers%
    
    \SetKwInOut{Input}{Input}
    \SetKwInOut{Output}{Output}
    \SetKwInOut{Parameters}{Parameters}
    \SetKwInOut{Variables}{Variables}
    
    \Parameters{\textit{nbuckets}: Number of buckets of the new hash table. \\
    \textit{hash}: User-specified hash function of the new hash table.}
    \BlankLine

    \Fn{void ht\_rebuild(ht *htp, \emph{nbuckets}, \emph{hash})}{
        \textbf{If} \textit{( trylock(rebuild\_lock) != SUCCESS )}  \textbf{return} -EBUSY\label{alg.rebuild.trylock}\;
        \textbf{If} \textit{( ! rebuild\_is\_required() )} \textbf{return} -EPERM\label{alg.rebuild.doublecheck}\;
        htp\_new := ht\_alloc(nbuckets, hash)\label{alg.rebuild.alloc}\;
        htp->ht\_new := htp\_new\label{alg.rebuild.assign1}\;
        \tcc{Wait for operations not aware of \emph{htp\_new}.}
        synchronize\_rcu()\label{alg.rebuild.sync1}\;

        \For{each bucket \emph{htbp} in htp\label{alg.rebuild.for1}} {
	    \For{each node \emph{htnp} in \emph{htbp}\label{alg.rebuild.for2}}{
                rebuild\_cur := htnp\label{alg.rebuild.cur1}\;
                \textbf{smp\_wmb()}\label{alg.rebuild.mb1}\;
		        key := htnp->key\;
                \If{lflist\_delete(htbp, key, IS\_BEING\_DIST)  != SUCCESS\label{alg.rebuild.del}\label{alg.rebuild.inner1}}{
                    \textbf{continue}\label{alg.rebuild.cont}\;
                }

                prepare\_node(htnp)\label{alg.rebuild.prepare}\;
                
                htbp\_new := htp\_new->bkts[htp\_new->hash(key)]\;
                \If{lflist\_insert(htbp\_new, htnp) != SUCCESS\label{alg.rebuild.insert}}{
                    call\_rcu(htnp, free)\label{alg.rebuild.freenode}\;
                }\label{alg.rebuild.inner2}
                \textbf{smp\_wmb()}\label{alg.rebuild.mb2}\;
                rebuild\_cur := NULL\label{alg.rebuild.cur2}\;
            }\label{alg.rebuild.for2end}
        }\label{alg.rebuild.for1end}
        \tcc{Wait for operations accessing nodes via \emph{htp->bks[]}.}
        synchronize\_rcu()\label{alg.rebuild.sync2}\;

        htp\_tmp := htp; \hspace{0.1cm} htp := htp\_new\label{alg.rebuild.assign2}\;

        \tcc{Wait for operations referencing to old hash table.}
        synchronize\_rcu()\label{alg.rebuild.sync3}\;
        unlock(rebuild\_lock)\label{alg.rebuild.unlock}\;
        free(htp\_tmp)\label{alg.rebuild.freeht}\;
        \textbf{return} \textit{SUCCESS}\label{alg.rebuild.return}\;
    }

    \caption{Rebuild operation of \sns.}
    \label{alg.rebuild}
\end{algorithm2e}

Function \emph{ht\_rebuild} traverses the old hash table,
and one-by-one distributes nodes to the new hash table
(Lines \ref{alg.rebuild.cur1}--\ref{alg.rebuild.cur2}).
For each node, the global variable \emph{rebuild\_cur} first points to
the node on line \ref{alg.rebuild.cur1}.
The two write barriers on lines \ref{alg.rebuild.mb1} and \ref{alg.rebuild.mb2}
pair with the read barriers in \emph{ht\_lookup} and \emph{ht\_delete}.
They together guarantee that the updates performed by \emph{ht\_rebuild}
to \emph{rebuild\_cur} and the two hash tables
can be seen by other operations in the same order.
Note that, for ease of presentation, we omit memory order specifications in the pseudo code.
In practice, all accesses to bucket pointers (e.g., \emph{htbp}), node pointers (e.g., \emph{htnp}),
and \emph{rebuild\_cur} 
must be made with the specifications of std::memory\_order\_acquired or release \cite{mmspec}.

Line \ref{alg.rebuild.del} deletes the node from the old hash table.
Function \emph{lflist\_delete} takes a third parameter \emph{IS\_BEING\_DISTRIBUTED},
indicating that the node with the matching key will be deleted from the old hash table,
but its memory will not be reclaimed.
If this delete operation fails,
which implies that the node has been deleted by other concurrent delete operations 
since the reference to the node was fetched on line \ref{alg.rebuild.for2},
the rebuild process skips this node (line \ref{alg.rebuild.cont}).
Line \ref{alg.rebuild.prepare} prepares the node for reuse by, for example,
cleaning the \emph{IS\_BEING\_DISTRIBUTED} bit of the node.
Then, line \ref{alg.rebuild.insert} inserts it into the proper bucket of the new hash table.
Note that if the insertion operation fails,
which means that one other node with the same key value
has been inserted into the new hash table by other threads,
line \ref{alg.rebuild.freenode} invokes \emph{call\_rcu}
which frees the node after currently unfinished operations 
referencing to this node have completed.
After the node has been inserted into the new hash table,
the global pointer \emph{rebuild\_cur} is set back to \emph{NULL}. 

After distributing all node of the old hash table,
line \ref{alg.rebuild.sync2} waits for unfinished common operations,
which may still hold references to the old hash table, to complete.
Line \ref{alg.rebuild.assign2} installs the new hash table as the current one,
and again line \ref{alg.rebuild.sync3} waits for all unfinished operations.
Then, line \ref{alg.rebuild.unlock} releases the global lock,
line \ref{alg.rebuild.freeht} frees the old hash table,
and finally line \ref{alg.rebuild.return} returns success.

For each iteration, \emph{ht\_rebuild()} 
deletes a node from the old hash table and then
inserts it into the new hash table,
reusing the node's memory.
One potential issue with the reuse of nodes is that it may 
redirect concurrent lookup operations that are traversing the old hash table
to the wrong lists.
For example, suppose that a lookup operation is traversing a hash bucket of the old hash table
and is referencing to node $\alpha$.
At this time, the rebuild operation distributes $\alpha$ 
by inserting $\alpha$ into the proper hash bucket in the new hash table.
This can redirect the lookup operation to the linked list in the new hash table,
and result in a false negative if the node with the matching key
is at the bottom of the linked list in the old hash table.
There are two approaches to overcoming this problem.
(1) The last nodes of the lists store corresponding bucket id's.
Once a lookup operation reaches the last node of a specified bucket,
it read the id value from the node,
and starts over if the value is not what expected.
(2) The lookup operation checks if $\alpha$ has been deleted
before moving forward to subsequent nodes.
For the lock-free linked list algorithm presented in this paper,
deleting a node is performed by setting the least-significant two bits of its \emph{next} field, 
such that the two steps (checking deletion and moving forward) 
can be performed atomically by using one \vn{compare-and-swap} operation
on the \emph{next} field of node $\alpha$.
The generic hash table in the Linux kernel \cite{rhashtable} uses the first approach, and \sn uses the second.



\subsection{Lookup operation}\label{sec.arch.impl.lookup}

The lookup operation is presented in Algorithm \ref{alg.lookup}.
The function first searches for the specified \emph{key} in the proper bucket
of the old hash table (line \ref{alg.lookup.find1}).
If a node with the matching key can be found in the bucket,
a pointer referencing to the node is returned (line \ref{alg.lookup.return1}).
Otherwise, line \ref{alg.lookup.checkresize} checks whether a rebuild operation is in progress.
If rebuild operations are absent,
line \ref{alg.lookup.return2} returns \emph{-ENOENT} indicating that
no node with the matching key can be found in \sns.
The two read barriers on lines \ref{alg.lookup.mb1} and \ref{alg.lookup.mb2}
pair with the two write barriers in \emph{ht\_rebuild}.
Line \ref{alg.lookup.global} continues the lookup operation
by checking the node pointed to by the global pointer \emph{rebuild\_cur}.
Recall that \emph{rebuild\_cur} always points to the node
that is currently in \emph{hazard period}.
If the node pointed to by \emph{rebuild\_cur} matches,
and if the LOGICALLY\_REMOVED bit of the \emph{next} field of the node has not been set,
which means that the node has not been deleted by concurrent delete operations,
line \ref{alg.lookup.return22} returns a pointer to the node.
Otherwise, function \emph{lookup} continues by searching the new hash table
and returns the pointer to the node if the \emph{lflist\_find} operation succeeds
(line \ref{alg.lookup.find2}).

\begin{algorithm2e}[ht]
    \footnotesize
    \SetStartEndCondition{ (}{)}{)}\SetAlgoBlockMarkers{}{\}}%
    \SetKwProg{Fn}{}{ \{}{}\SetKwFunction{FRecurs}{void FnRecursive}%
    \SetKwFor{For}{for}{ \{}{}%
    \SetKwIF{If}{ElseIf}{Else}{if}{ \{}{elif}{else \{}{}%
    \SetKwFor{While}{while}{ \{}{}%
    \SetKwRepeat{Repeat}{repeat\{}{until}%
    \AlgoDisplayBlockMarkers%

    \SetKwInOut{Input}{Input}
    \SetKwInOut{Output}{Output}
    \SetKwInOut{Parameters}{Parameters}
    \SetKwInOut{Variables}{Variables}
    
    \textbf{Local variables: } struct snapshot ss;   struct node *cur, *htbp, *htbp\_new\;
    \Fn{node *ht\_lookup(ht *htp, long key)} {
        htbp := htp->bkts[htp->hash(key)]\label{alg.lookup.bkt1}\;
        \textbf{if} (\textit{lflist\_find(htbp, key, \&ss) = \emph{SUCCESS}}) \label{alg.lookup.find1}
            \textbf{\{} \textbf{return} ss.cur\label{alg.lookup.return1}\label{alg.lookup.find1end} \textbf{\}}\;
        
        \textbf{if} (\textit{htp->ht\_new = NULL\label{alg.lookup.checkresize}})
            \textbf{\{} \textbf{return} -ENOENT\label{alg.lookup.return2}\label{alg.lookup.checkresizeend} \textbf{\}}\;
        \textbf{smp\_rmb()}\label{alg.lookup.mb1}\;
        cur := rebuild\_cur\label{alg.lookup.assignrebuild}\;
        \If{cur and (cur->key = key) and !logically\_removed(cur)\label{alg.lookup.global}}{
            \textbf{return} cur\label{alg.lookup.return22}\;
        }
        \textbf{smp\_rmb()}\label{alg.lookup.mb2}\;
        htp\_new := htp->ht\_new\;
        htbp\_new := htp\_new->bkts[htp\_new->hash(key)]\;
        \textbf{if} (\textit{lflist\_find(htbp\_new, key, \&ss) = \emph{SUCCESS}}) \label{alg.lookup.find2}
            \textbf{\{} \textbf{return} ss.cur\label{alg.lookup.return3} \textbf{\}}\;
        \textbf{else} \textbf{return} -ENOENT\label{alg.lookup.return4}\;

    }
    
    \caption{Lookup operation of \sns.}
    \label{alg.lookup}
\end{algorithm2e}


Algorithm \ref{alg.lookup} shows that
a lookup() operation first searches for the node with the matching key 
in the old hash table (Line \ref{alg.lookup.find1}),
then checks if the node pointed to by \emph{rebuild\_cur} is the right node
(Line \ref{alg.lookup.global}), 
and finally searches in the new hash table (Line \ref{alg.lookup.find2}).
This manipulation order guarantees that lookup operations
can always find the node even if a rebuild operation is in progress.
That is, the following lemma holds:

\begin{lemma}
    If \sn contains node $\alpha$ with key value of \vn{K},
    operation ht\_lookup(\vn{K}) can return a pointer to $\alpha$,
    no matter if a rebuild operation is in progress.
    \label{lemma.lookup}
\end{lemma}
\begin{proof}
    Obviously, if rebuild operation is absent,
    node $\alpha$ resides in the only hash table.
    Operation \emph{ht\_lookup(\vn{K})} can find the node
    in the only hash table 
    (lines \ref{alg.lookup.bkt1} - \ref{alg.lookup.find1end}).

    We then prove that \emph{ht\_lookup(\vn{K})} can find the node
    when a rebuild operation is in progress.
    The code snippet to distribute a node is
    shown on lines \ref{alg.rebuild.cur1} - \ref{alg.rebuild.cur2}.
    We use $write_{rebuild}(rebuild\_cur, \alpha)$ to denote the event
    in which the thread running the rebuild operation 
    (henceforth rebuild thread for short)    
    assigns the address of node $\alpha$
    to the global variable \emph{rebuild\_cur} (line \ref{alg.rebuild.cur1}),
    and use $delete_{rebuild}(old, \alpha)$ and $insert_{rebuild}(new, \alpha)$
    to denote the events
    in which node $\alpha$ is deleted from and inserted into
    the old and the new hash table, respectively
    (lines \ref{alg.rebuild.del} and \ref{alg.rebuild.insert}).
    Similarly,
    we use $find_{lookup}(old, \alpha)$ and $find_{lookup}(new, \alpha)$
    to denote the events in which the \emph{lookup} thread
    searches for node $\alpha$ in the old and the new hash table, respectively
    (lines \ref{alg.lookup.find1} and \ref{alg.lookup.find2}).
    We use $find_{lookup}(rebuild\_cur, \alpha)$ to denote the event
    in which the lookup thread checks the node pointed to by \emph{rebuild\_cur} (line \ref{alg.lookup.global}). 
    In the following proof,
    since the rebuild thread is the only thread that performs
    write/delete/insert operations,
    and the lookup thread is the only thread that performs
    find operations, 
    we omit thread symbol without introducing any ambiguity.
    For brevity, we use the acronym \emph{rbc} to stand for \emph{rebuild\_cur}.
    One event $e_1$ precedes another event $e_2$, 
    written $e_1 \prec e_2$ , 
    if $e_1$ occurs at an earlier time.

    By inspecting the code of \emph{ht\_rebuild} in Algorithm \ref{alg.rebuild}
    we get that:
\begin{equation}
    \footnotesize
    write(rbc, \alpha) \prec delete(old, \alpha) \prec insert (new, \alpha) 
        \prec write(rbc, \perp)
    \label{equa.rebuild}
\end{equation}

    By inspecting the code of \emph{ht\_lookup} in Algorithm \ref{alg.lookup},
    we get that:
\begin{equation}
    \footnotesize
    find(old, \alpha) \prec find(rbc, \alpha) \prec find(new, \alpha)
    \label{equa.lookup}
\end{equation}

When the lookup and the rebuild thread
is simultaneously accessing node $\alpha$,
there are three types of interleaving
between these two threads:

\begin{itemize}[leftmargin=*]
    \item
    {\footnotesize $find(old, \alpha) \prec delete(old, \alpha)$},
    which implies that the lookup thread searches for node $\alpha$
    before the rebuild thread starts distributing the node.
    Thus, the node can be found in the old hash table and
    the lookup operation
    can return a pointer to $\alpha$ on line \ref{alg.lookup.return1}.

    \item 
    {\footnotesize $insert(new, \alpha) \prec find(new, \alpha)$},
    which implies that the lookup thread searches for node $\alpha$
    after it has been inserted into the new hash table
    by the rebuild thread. 
    Thus, the node can be found in the new table and
    the lookup operation
    can return a pointer to $\alpha$ on line \ref{alg.lookup.return3}.

    \item
    {\footnotesize $delete(old, \alpha) \prec find(old, \alpha) \prec \dots
        \prec find(new, \alpha) \prec insert(new, \alpha)$},
    which implies that the \emph{lookup} thread searches for node $\alpha$
    which is in \emph{hazard period}.
    Combined with Equations \ref{equa.rebuild} and \ref{equa.lookup},
    we get the following event sequence:

    {\footnotesize $write(rbc, \alpha) \prec delete(old, \alpha) \prec find(old, \alpha) \prec
        find(rbc, \alpha) \prec \\ find(new, \alpha) \prec insert(new, \alpha)
        \prec write(rbc, \perp)$}

    It follows that:

    {\footnotesize $write(rbc, \alpha) \prec find(rbc, \alpha)
        \prec write(rbc, \perp)$}

    Once \emph{rbc(rebuild\_cur)} is set to point to node $\alpha$ it remains.
    Hence the \emph{lookup} thread can find node $\alpha$ via \emph{rebuild\_cur} 
    and can return a pointer to it on line \ref{alg.lookup.return22}.
\end{itemize}

In overall, if there is a node with the matching key in \sns,
it is guaranteed that the \emph{ht\_lookup} operation can
find the node and return a pointer to it,
no matter if a rebuild operation is in progress.
\end{proof}



\subsection{Delete operation}\label{sec.arch.impl.delete}

The delete operation of \sn is shown in Algorithm \ref{alg.delete}.
The function first attempts to delete the node from the old hash table
on line \ref{alg.delete.del1},
and returns \emph{SUCCESS} if succeeds on line \ref{alg.delete.return1}.
Otherwise,
the function continues by checking if a rebuild operation is in progress
on line \ref{alg.delete.checknew}.
The two read barriers on lines \ref{alg.delete.mb1} and \ref{alg.delete.mb2}
pair with the two write barriers in \emph{ht\_rebuild}.
If a rebuild operation is in progress,
the delete operation checks if the node
pointed to by \emph{rebuild\_cur} has the expected key value
on line \ref{alg.delete.checkglobal},
and if the answer is yes, the rebuild operation deletes the node
by setting the \emph{LOGICALLY\_REMOVED} bit of the \emph{next} field of the node
(Line \ref{alg.delete.deleteglobal}).
Function \emph{delete()} continues by attempting to delete the node
with the matching key from the new hash table (Line \ref{alg.delete.del2}).
If the delete operation fails,
line \ref{alg.delete.return4} returns \emph{-ENOENT} indicating that
no node with the matching key can be found in \sns.

\begin{algorithm2e}[ht]
    \footnotesize
    \SetStartEndCondition{ (}{)}{)}\SetAlgoBlockMarkers{}{\}}%
    \SetKwProg{Fn}{}{ \{}{}\SetKwFunction{FRecurs}{void FnRecursive}%
    \SetKwFor{For}{for}{ \{}{}%
    \SetKwIF{If}{ElseIf}{Else}{if}{ \{}{elif}{else \{}{}%
    \SetKwFor{While}{while}{ \{}{}%
    \SetKwRepeat{Repeat}{repeat\{}{until}%
    \AlgoDisplayBlockMarkers%
    
    \SetKwInOut{Input}{Input}
    \SetKwInOut{Output}{Output}
    \SetKwInOut{Parameters}{Parameters}
    \SetKwInOut{Variables}{Variables}
    
    \textbf{Local variables}: struct node *cur, *htbp, *htbp\_new\;
    \Fn{int ht\_delete(ht *htp, long key)} {

        htbp = htp->bkts[htp->hash(key)]\label{alg.delete.norebuild}\;
        \If{lflist\_delete(htbp, key, LOGICALLY\_REMOVED) = SUCCESS\label{alg.delete.del1}} {
            \textbf{return} SUCCESS\label{alg.delete.return1}\;
        }\label{alg.delete.norebuild2}
        htp\_new := htp->ht\_new\;
        \textbf{if} (\textit{htp\_new = NULL}) \label{alg.delete.checknew}
            \textbf{\{} \textbf{return} -ENOENT\label{alg.delete.return2}\label{alg.delete.checknewend} \textbf{\}}\;

        \textbf{smp\_rmb()}\label{alg.delete.mb1}\;

        cur := rebuild\_cur\label{alg.delete.assignrebuild}\;
        \If{cur and (cur->key = key)\label{alg.delete.checkglobal}}{
	    \emph{set\_flag}(cur, LOGICALLY\_REMOVED)\label{alg.delete.deleteglobal}\; 
            \textbf{return} SUCCESS\label{alg.delete.returnX}\;
	    }\label{alg.delete.checkglobale}

        \textbf{smp\_rmb()}\label{alg.delete.mb2}\;

        htbp\_new := htp\_new->bkts[htp\_new->hash(key)]\;
        \If{lflist\_delete(htbp\_new, key, LOGICALLY\_REMOVED) = SUCCESS\label{alg.delete.del2}}{
            \textbf{return} SUCCESS\label{alg.delete.return3}\;
        }
        \textbf{return} -ENOENT\label{alg.delete.return4}\;
        
    }
    
    \caption{Delete operation of \sns.}
    \label{alg.delete}
\end{algorithm2e}

To delete a node, \sn adopts a classic lightweight mechanism presented in \cite{Michael2002Hashing},
by separating the deletion of a node into two stages:
\emph{logical} and \emph{physical} deletion.
The first stage is to mark a node 
(e.g., by setting the least-significant bits in the \emph{next} field)
to prevent subsequent lookup operations from returning this node,
and prevent subsequent insert and delete operations from inserting and deleting nodes
after this node.
The second stage, which is typically performed by subsequent lookup operations,
is to physically delete the node from the list by 
swinging the next pointer of the previous node to the next node in the list
and then reclaiming the node memory.

Since a delete operation fundamentally performs lookup operations
in addition to a \emph{logical deletion} if the node with the matching key is found,
it is straightforward to prove that if the manipulation order of a delete operation
is the same as that of a lookup operation shown in Algorithm \ref{alg.lookup},
the delete operation can always find the node (because of Lemma \ref{lemma.lookup})
and delete the node (a logical deletion can always succeed).
That is, the following lemma holds:
\begin{lemma}
    If \sn contains node $\alpha$ with the key value of \vn{K},
    operation ht\_delete(\vn{K}) can successfully delete node $\alpha$,
    no matter if a rebuild operation is in progress.
    \label{lemma.deletion}
\end{lemma}

Since the proof is similar to that of Lemma \ref{lemma.lookup},
due to lack of space, we omit the proof of Lemma \ref{lemma.deletion}
in this paper. 



\subsection{Insert operation}\label{sec.arch.impl.insert}

Function \emph{ht\_insert()} in Algorithm \ref{alg.insert}
inserts a new node into \sns.
The function first allocates a new node and initializes it
(line \ref{alg.insert.alloc})
and then checks if a rebuild operation is in progress
on line \ref{alg.insert.checkglobal}.

\begin{algorithm2e}[ht]
    \footnotesize
    \SetStartEndCondition{ (}{)}{)}\SetAlgoBlockMarkers{}{\}}%
    \SetKwProg{Fn}{}{ \{}{}\SetKwFunction{FRecurs}{void FnRecursive}%
    \SetKwFor{For}{for}{ \{}{}%
    \SetKwIF{If}{ElseIf}{Else}{if}{ \{}{elif}{else \{}{}%
    \SetKwFor{While}{while}{ \{}{}%
    \SetKwRepeat{Repeat}{repeat\{}{until}%
    \AlgoDisplayBlockMarkers%
    
    \SetKwInOut{Input}{Input}
    \SetKwInOut{Output}{Output}
    \SetKwInOut{Parameters}{Parameters}
    \SetKwInOut{Variables}{Variables}
    
    \textbf{Local variables}: struct node *htnp, *htbp, *htbp\_new;  struct ht *htp\_new\;
    \Fn{int ht\_insert(ht *htp, long key)} {
        htnp := allocate\_node(key)\label{alg.insert.alloc}\;
        htp\_new := htp->ht\_new\;
        
        \eIf{htp\_new = NULL\label{alg.insert.checkglobal}}{
        htbp := htp->bkts[htp->hash(key)]\label{alg.insert.common.b}\;
            \textbf{if} (\textit{lflist\_insert(htbp, htnp)}) \label{alg.insert.insertold}
            \textbf{\{} \textbf{return} SUCCESS\label{alg.insert.return1}\label{alg.insert.common.e} \textbf{\}}\;
	    } {\label{alg.insert.newtable.b}
            htbp\_new := htp\_new->bkts[htp\_new->hash(key)]\;
            \textbf{if} ( \textit{lflist\_insert(htbp\_new, htnp) }) \label{alg.insert.insertnew}
                \textbf{\{} \textbf{return} SUCCESS\label{alg.insert.return2} \textbf{\}}\;

	}\label{alg.insert.newtable.e}

        free(htnp)\label{alg.insert.free}\;
        \textbf{return} -EEXIST\label{alg.insert.return3}\;
    }
    
    \caption{Insert operation of \sns.}
    \label{alg.insert}
\end{algorithm2e}

If there is no rebuild operation in progress,
function \emph{ht\_insert} inserts the new node in the old hash table 
on line \ref{alg.insert.insertold}.
In contrast, if a rebuild operation is in progress,
it inserts the node in the new hash table on line \ref{alg.insert.insertnew}.
If any insertion fails, which implies that 
a node with the same key value has been inserted into \sn
before this insert operation is performed.
Function \emph{ht\_insert} frees the newly allocated node on line \ref{alg.insert.free}
and returns a failure message on line \ref{alg.insert.return3}.
Since the RCU technique is used to synchronize insert operations and rebuild operation,
the following lemma holds:
\begin{lemma}
    When a rebuild operation is in progress,
    function \emph{ht\_insert} can successfully insert node $\alpha$ into \sns.
    \label{lemma.insert.helper}
\end{lemma}
\begin{proof}
    Recall that function \emph{ht\_insert()} is run in an RCU read-side critical section,
    and function \emph{ht\_rebuild()} performs a \emph{synchronize\_rcu} barrier on line \ref{alg.rebuild.sync1}
    (called \emph{barrier 1})
    before distributing nodes to 
    the new hash table.
    If the insert operation starts before \emph{barrier 1},
    it may or may not see the new hash table, 
    and hence could insert the node into either the old or the new hash table.
    Inserting the node into any of the hash tables is correct,
    because \emph{barrier 1} prevents function \emph{ht\_rebuild()} from starting distributing nodes
    until the insert operation completes and leaves its RCU read-side critical section.
    In the other case, if the insert operation starts after \emph{barrier 1},
    which means that function \emph{ht\_rebuild()} is distributing nodes,
    the insert operation will insert the node into the new hash table.
    A second \emph{synchronize\_rcu} barrier on line \ref{alg.rebuild.sync2}
    force the function \emph{ht\_rebuild()} to
    wait until the insert operation completes and leaves its RCU read-side critical section.
\end{proof}

\noindent Now, we prove that the following lemma holds: 

\begin{lemma}
    No matter if a rebuild operation is in progress,
    when an operation \emph{ht\_insert(K)} returns,
    it is guaranteed that a node with the key value of \vn{K}
    can be found in the hash table.
    \label{lemma.insert}
\end{lemma}
\begin{proof}
    If a rebuild operation is absent, which is the common case,
    function \emph{ht\_insert} inserts the new node into the only hash table
    (Lines \ref{alg.insert.common.b}--\ref{alg.insert.common.e}).
    If a rebuild operation is in progress,
    Lemma \ref{lemma.insert.helper} guarantees that
    \emph{ht\_insert} will eventually insert the new node
    into the new hash table.
    Function \emph{ht\_insert} fails only if another node with the same key value 
    has been inserted into the hash table,
    which guarantees that a node with the key value of \vn{K} can be found.
\end{proof}

\section{Correctness}\label{sec.correcness}



For brevity, we provide only informal proof sketches.
The full proof of correctness will be provided in the full version of the paper.
Recall that \sn is modular.
Therefore, if a lock-free/wait-free set algorithm can provide the API set listed in Algorithm \ref{alg.structure.list},
\sn can utilize it as the implementation of hash buckets.
Therefore, the correctness of \sn depends on the set algorithm used.
In this paper, we choose the RCU-based linked list presented in Section \ref{sec.bb.list}
as the example.

\textbf{Safety}:
When rebuild operations are absent, safety is proved by following similar arguments
as those used to prove Michael's lock-free hash table \cite{Michael2002Hashing}.
When a rebuild operation is running, 
Lemmas \ref{lemma.lookup}, \ref{lemma.deletion}, and \ref{lemma.insert} show that
concurrent lookup, insert, and delete operations can execute correctly. 


\textbf{Progress guarantee}:
\sn is a blocking data structure because 
the RCU technique is used,
and therefore the rebuild operation can be blocked.
Specifically, \emph{ht\_rebuild} serializes concurrent rebuild requests by using a mutex lock
and waits for prior hash table operations by using the \emph{synchronize\_rcu} barriers.
This is acceptable for a practical implementation 
because rebuild operations commonly are rare and their speed is not the major concern
if they do not noticeably affect concurrent hash table operations.

Nevertheless, the lookup, insert, and delete operations
could be lock-free/wait-free, which is determined by the set algorithm used.
For example, Algorithm \ref{alg.lookup} shows that a lookup operation
invokes the list operation \emph{lflist\_find} twice.
Other statements in a lookup operation are regular instructions,
which can complete in a finite number of CPU cycles.
As a result, for the implementation of \sn presented in this paper,
its lookup operation is lock-free because searching a linked list is lock-free.
(The find operation of Michael's list may start over from the list head when they find a marked node.)
Similarly, we can prove that the insert and delete operations of \sn are lock-free.
(As discussed in Section \ref{sec.bb.list}, \emph{call\_rcu} is used in reclaiming deleted nodes,
such that delete operations will not block.)
Note that since \sn is modular,
programmers can instead use a wait-free linked list 
and make the common operations become wait-free.



\textbf{Linearizability}:
\sn is linearizable if the set algorithm used is linearizable.
The linked list presented in Section \ref{sec.bb.list} is linearizable,
because we did not change the control flow of Michael's list algorithm.
Specifically, we keep all the CAS instructions and memory barriers that the algorithm contains.
As a result, \sn is linearizable because every operation on the hash table
has a specific linearization point, where it takes effects.

Specifically,
every lookup operation that finds the node with the matching key
via \emph{rebuild\_cur}
takes effect on line \ref{alg.lookup.assignrebuild}.
For other cases, the lookup operation linearizes in either of the two invocations of \emph{lflist\_find}.
Similarly,
every delete operation that finds the node with the matching key
via \emph{rebuild\_cur}
takes effect on line \ref{alg.delete.assignrebuild}.
For other cases, the delete operation linearizes in either of the two invocations of \emph{lflist\_delete}.
Every insert operation takes effect in either of the two invocations of
\emph{lflist\_insert}.

\section{Evaluations}\label{sec.eva}

In this section, we demonstrate that on three different architectures
(1) the overall performance of \sn matches or slightly exceeds
other practical alternatives,
(2) \sn noticeably outperforms other alternatives under heave workloads,
and (3) the rebuild operation of \sn is efficient and predictable in execution time.


\subsection{Evaluation Methodology} \label{sec.eva.meth}
We choose Xu's hash table \cite{Herbert2010Hashing} (\emph{HT-Xu} for short)
as the representative of dynamic hash tables
that maintain two sets of list pointers in each node.
We choose the \emph{rhashtable} algorithm in the Linux kernel
\cite{rhashtable} (\emph{HT-RHT} for short)
as the representative of dynamic hash tables
that use a single set of list pointers. 
We also compare the performance of \sn to the famous split-ordered-list 
resizable hash tables \cite{Shalev2006SplitorderedLL} (\emph{HT-Split} for short)
that maintain a single ordered linked list for a hash table.

We implemented \sn as a user-space library in C.
The original implementation of \emph{HT-Xu} is tightly combined with
the multicasting code of the Linux kernel,
so we use the implementation in perfbook \cite{perfbook},
which is a good representative of \emph{HT-Xu} and run in user-space.
We implemented a user-space \emph{HT-RHT} that is strictly close to
the original kernel implementation except that
we omitted some sophisticated features such as Nested Tables to handle 
GFP\_ATOMIC memory allocation failures and 
Listed Tables to support duplicated nodes.
The open-source project \emph{userspace-rcu} \cite{userspace-rcu}
includes a up-to-date implementation of \emph{HT-Split}.
Hence, we use the implementation in \emph{userspace-rcu} in experiments.

For all of the implementations, optimizations such as cache-line padding are applied
if possible.
We compile the code with GCC 5.4.0 on all platforms where Ubuntu 16.04.5 is installed.
We use -O3 as our optimization level without any special optimization flags.

\textbf{Hardware platforms }
We evaluated the performance of aforementioned hash tables
on three different architectures.
Table \ref{table.platforms} lists the key characteristics of these servers.

\begin{table}[ht]
        \centering
        \footnotesize
        \caption{Summary of experimental platforms.}
        \label{table.platforms}
        \begin{tabular}{lccccr}
                \hline\hline\noalign{\smallskip}
		Processor Model & Speed & \#Sockets & \#Cores & LLC & Memory\\
                \noalign{\smallskip}
                \hline
		Intel Ivy Bridge & 2.6 G & 2 & 24 & 15 M & 64 G\\
		IBM Power9 & 2.9 G & 1 & 16 & 80 M & 16 G\\
		Cavium ARMv8 & 2.0 G & 2 & 96 & 16 M & 32 G\\
                \hline
        \end{tabular}
\end{table}

\textbf{Benchmarking framework}
To compare the performance and robustness of \sn to the alternatives,
we extended the \emph{hashtorture} benchmarking framework in \cite{perfbook}.
Specifically, the extended framework consists of a specified number of
\emph{worker threads}, each of which performs the workload with
the specified distribution of \emph{insert}, \emph{delete}, and \emph{lookup} operations
specified by parameter $m$.
In mapping worker threads to CPU cores,
we use a \emph{performance-first} mapping;
a new thread is mapped to the CPU core
that has the smallest number of worker threads running on it.
Experiments performed on a single CPU socket are marked with an \emph{*}, 
experiments performed on multiple CPU sockets are marked with a \emph{\#}, 
and experiments in which worker threads oversubscribe CPU cores are marked with a \emph{!}.
In experiments,
we varied parameters that significantly affects the performance of concurrent hash tables:
the mix of operations $m$, the average load factor $\alpha$, 
the number of buckets $\beta$,  and the range of keys $U$.
When a test starts,
every worker thread performs an infinite loop.
In each iteration, the worker thread randomly selects an operation type
(\emph{insert}, \emph{delete}, or \emph{lookup}) according to the specified distribution $m$,
chooses a key from 0 to the specified upper bound $U$,
and then performs the specified operation.
We chose parameters as follows:
$U$ is set to ten million
that is large enough to prevent CPU caches from buffering the whole test set.
We controlled the average load factor $\alpha$
indirectly by inserting $\alpha*\beta$ nodes in a hash table
before starting a test, and by selecting the ratio of \emph{insert}
to be equal to that of \emph{delete},
which guarantees a fixed number of nodes in the hash table.



\begin{figure}[h]
    \centering
    \subfloat[90\% lookup. Load factor is set to 2.]{
        \includegraphics[scale=0.17, angle=-90]{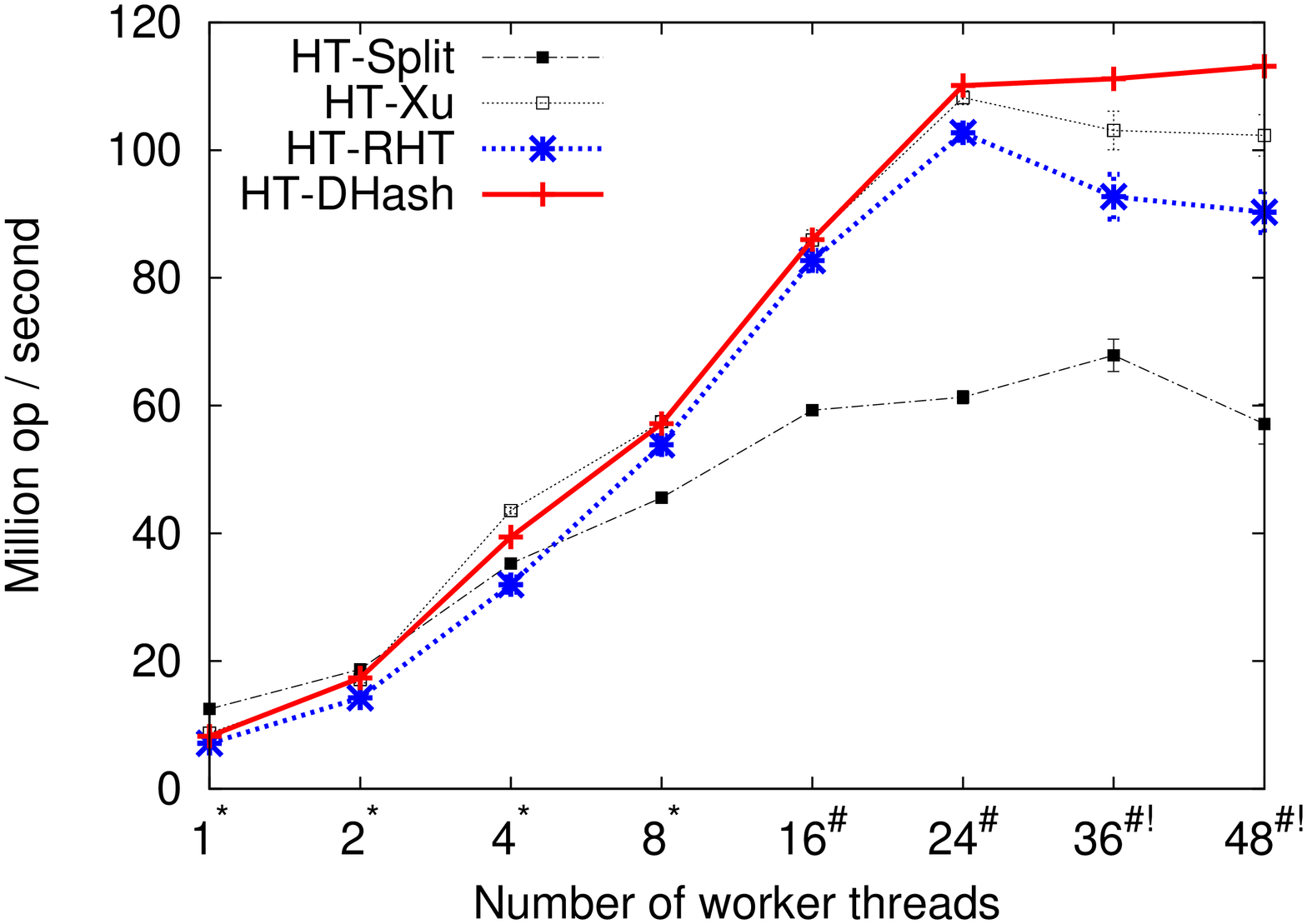}
        \label{fig.perf.LD2.5}
    }
    \subfloat[80\% lookup. Load factor is set to 2.]{
        \includegraphics[scale=0.17, angle=-90]{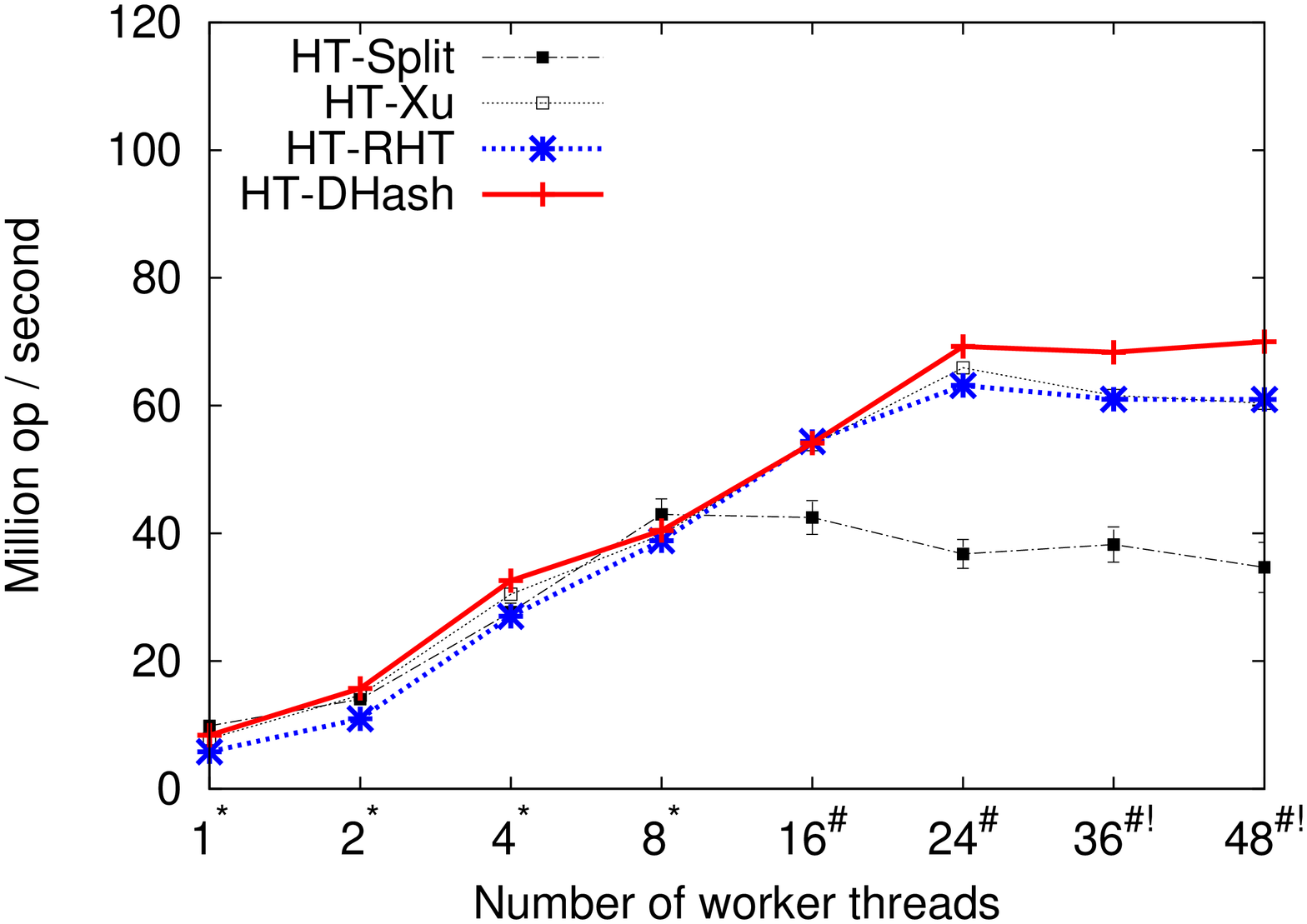}
        \label{fig.perf.LD2.10}
    }

    \subfloat[90\% lookup. Load factor is set to 20.]{
    \includegraphics[scale=0.17, angle=-90]{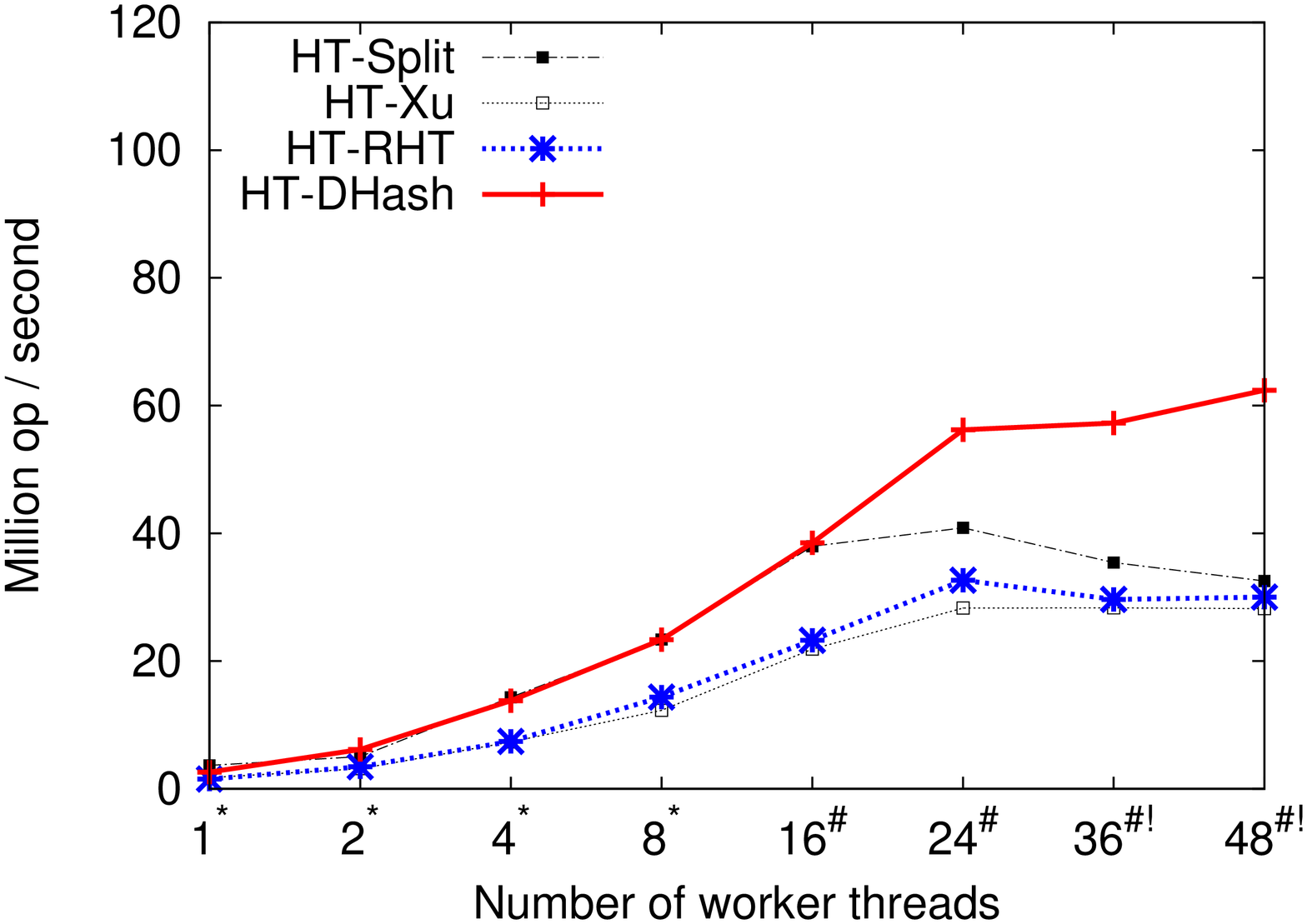}
    \label{fig.perf.LD20.5}
    }
    \subfloat[80\% lookup. Load factor is set to 20.]{
    \includegraphics[scale=0.17, angle=-90]{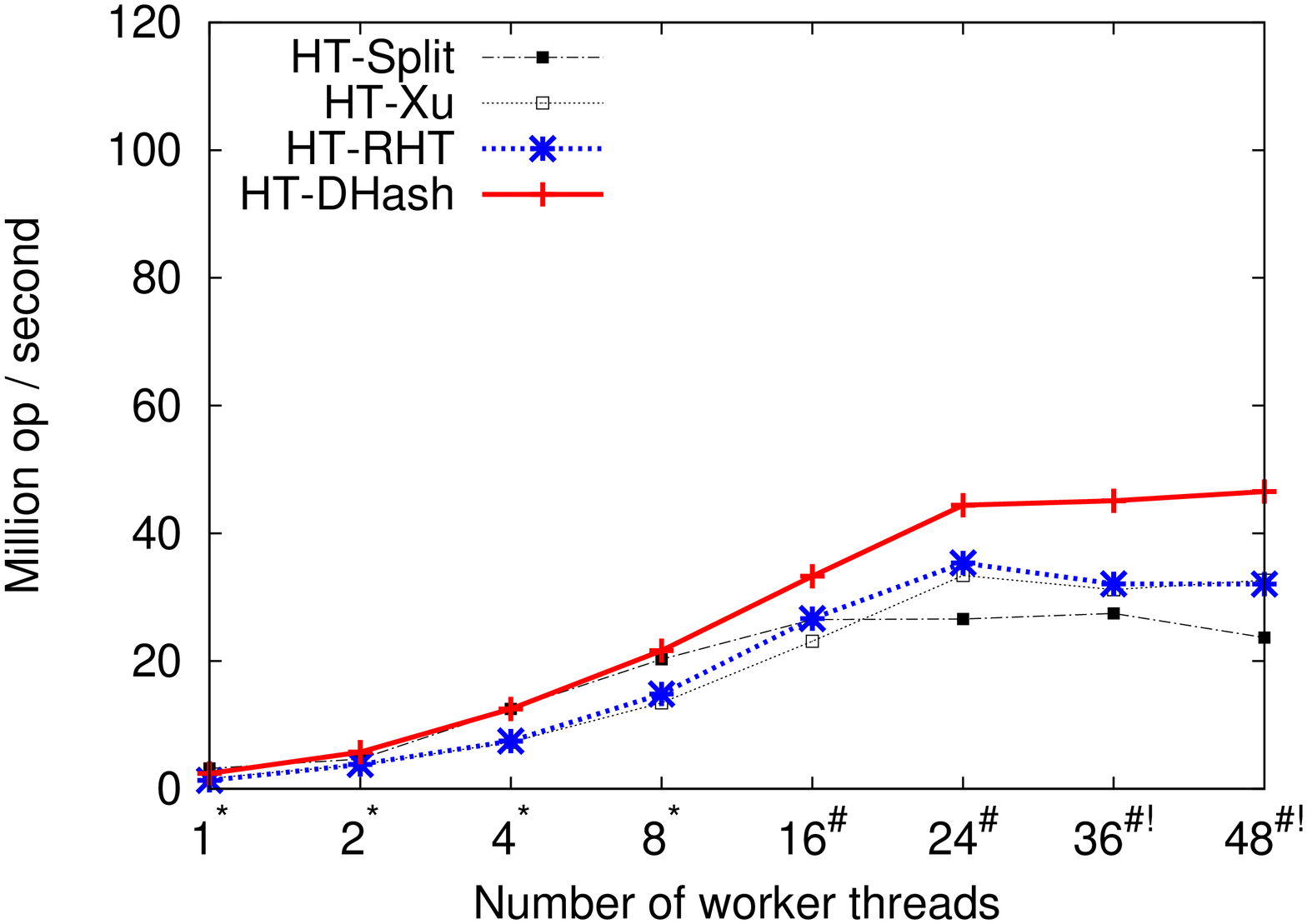}
    \label{fig.perf.LD20.10}
    }

    \subfloat[90\% lookup. Load factor is set to 200.]{
    \includegraphics[scale=0.17, angle=-90]{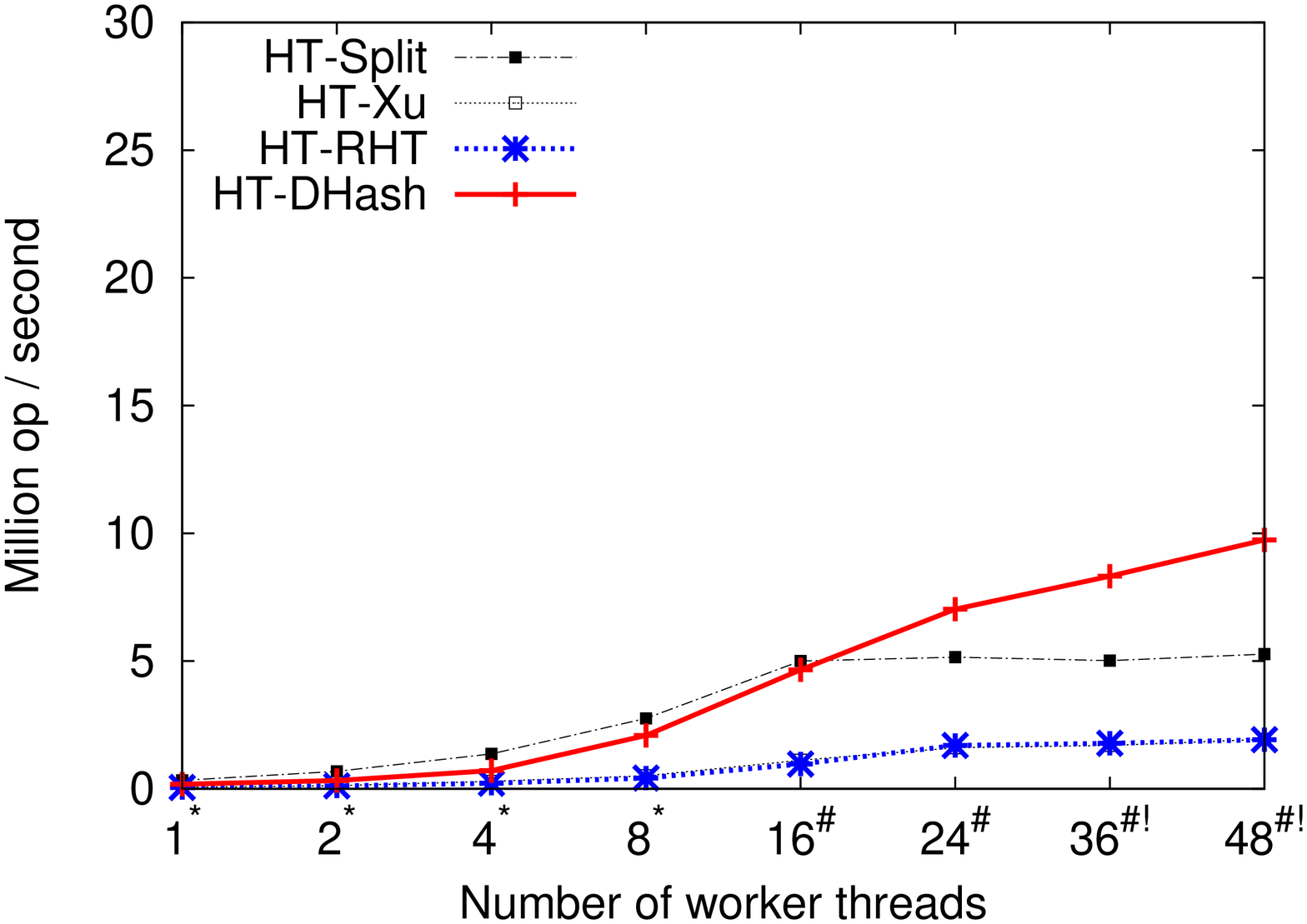}
    \label{fig.perf.LD200.5}
    }
    \subfloat[80\% lookup. Load factor is set to 200.]{
    \includegraphics[scale=0.17, angle=-90]{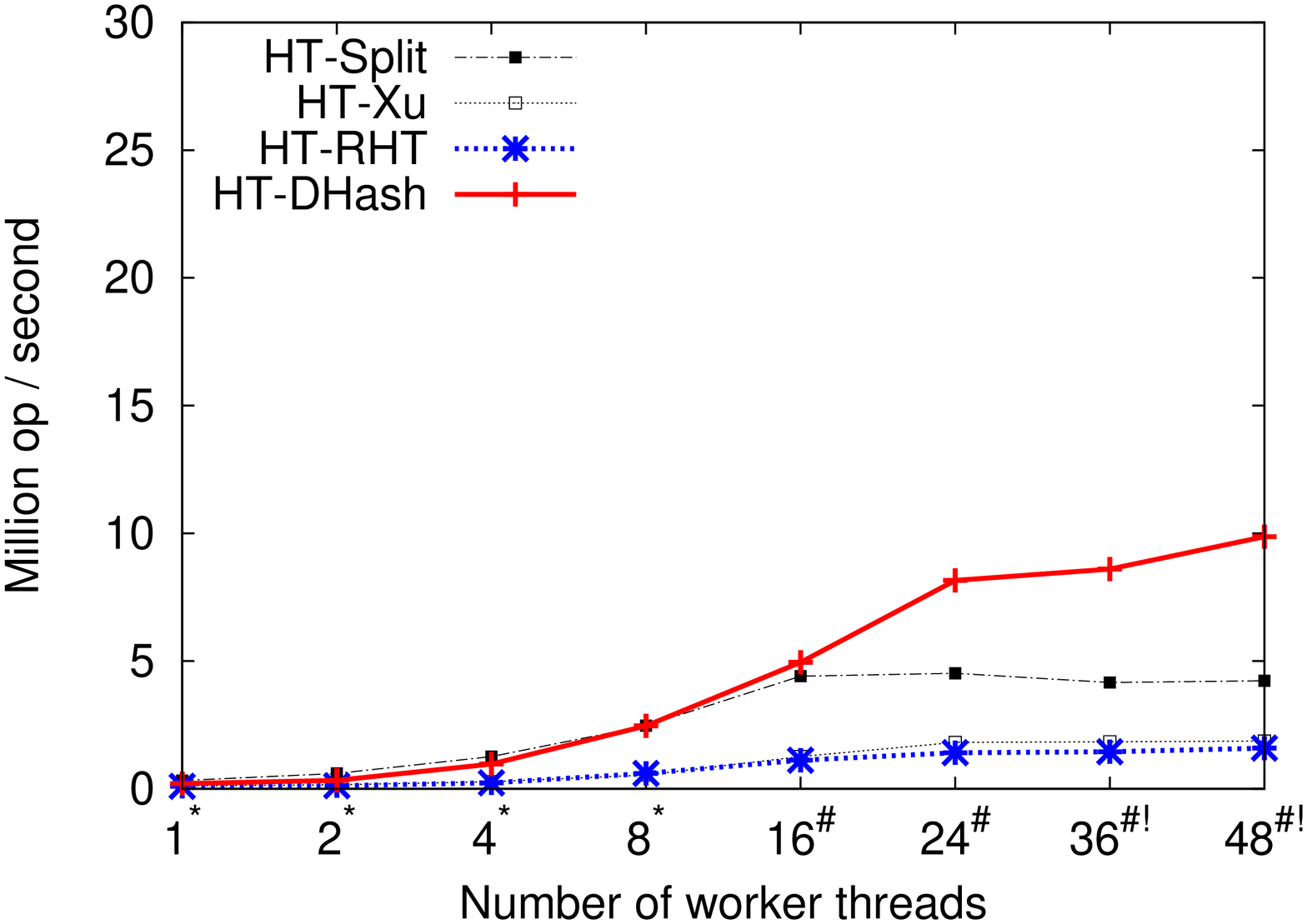}
    \label{fig.perf.LD200.10}
    }
    \caption{Performance of hash tables on Intel Ivy Bridge.}
    \label{fig.perf}
\end{figure}

\subsection{Overall performance}\label{sec.eva.perf}

Figures \ref{fig.perf} shows
the overall performance of the hash tables with various load factors and operation mixes.
We only present in this section the results for representative experiments
performed on the Intel Ivy Bridge server.
Experimental results on other architectures are discussed in Section
\ref{sec.eva.arch}.
Note that for clarity, the range of the y-axis in the last two figures 
is smaller than in other Figures.
Standard deviation is denoted by vertical bars,
which may be too small to be visible in the figure.
To compare with \emph{HT-Split}, in this experiment,
both the new and the old hash tables use the same hash function,
degrading \sns, \emph{HT-RHT}, and \emph{HT-Xu} to resizable hash tables.
A rebuild thread continuously rebuilds a hash table from its initial size to
the alternative size and back.
While continuos resizes do not necessary reflect a common usage pattern for a hash table,
this experiment noticeably demonstrate the overall performance of a hash table
under rebuilding/resizing, demonstrating the baseline performance of the hash table.

Experimental results show that 
(1) the overall performance of \sn matches or slightly exceeds other practical alternatives
with small average load factors (Figures \ref{fig.perf.LD2.5}--\ref{fig.perf.LD2.10}),
and (2) \sn significant outperforms other hash tables
under heavy workloads (Figures \ref{fig.perf.LD200.5}--\ref{fig.perf.LD200.10}).
For example,
Figure \ref{fig.perf.LD200.10} shows that
when 48 worker threads concurrently execute operations,
\sn can still handle 9.87 million operations per second,
which outperforms \emph{HT-Split}, \emph{HT-Xu}, and \emph{HT-RHT}
by factors of 2.3, 5.3, and 6.2, respectively.

Another important observation is that \sn outperforms other alternatives
with respect to scalability and robustness.
Figures \ref{fig.perf} shows that
when the number of worker threads exceed the number of CPU cores (24 for the Intel Ivy Bridge server),
the performance of  \sn increases slightly
despite the fact that more threads is contending the hash table.
For example, Figure \ref{fig.perf.LD200.5} shows that
as the number of worker threads increases from 24 to 48,
the overall performance of \sn increases from 7.03 to 9.74 million operations per second.
The performance of other alternatives, however, becomes flat or decreases
due to the increased contention on bucket locks.


\subsection{Rebuilding efficiency}\label{sec.eva.rebuild}

In this section, we measure how fast various rebuild operations can rebuild.
For brevity, Figure \ref{fig.eva.rebuild} only shows the results of representative experiments
running on the Intel Ivy Bridge server, and with one worker thread.
The x-axis of the figure is the amount of nodes in hash tables,
and the y-axis the time spent in rebuilding these hash tables.
Note that for clarity, both axes do not start from zero,
and the y-axis is shown as log scale.
The results of experiments with
different percentages of lookup operations (90\% and 80\%) are shown in Figures
\ref{fig.rebuild.LD.5.5.90} and \ref{fig.rebuild.LD.33.33.34},
respectively.

We make the following observation.
As expected, the cost of the resize operation of \emph{HT-Split} is consistently low,
because of the fact that \emph{HT-Split} is a resizable hash table
and hence it only changes the array of bucket pointers
when resizing.
The rebuild operation of \emph{HT-Xu} is much more efficient
compared with \sn and \emph{HT-RHT} because of its two-sets-of-pointers property,
which allows a rebuild operation to rebuild the hash table by
traversing the hash table once.
For other dynamic hash tables, which need to distribute all nodes to the new hash table,
the time required is basically linear to the amount of nodes in the old hash table.
For this type of hash tables,
\sn outperforms \emph{HT-RHT} in rebuilding efficiency because
\emph{HT-RHT} always traverses a bucket list and then distributes the last node.
In contrast, \sn distributes the head nodes, avoiding the traversing overheads.

\begin{figure}[!h]
    \centering
    \subfloat[5\% Insert, 5\% Delete]{
        \includegraphics[scale=0.16, angle=-90]{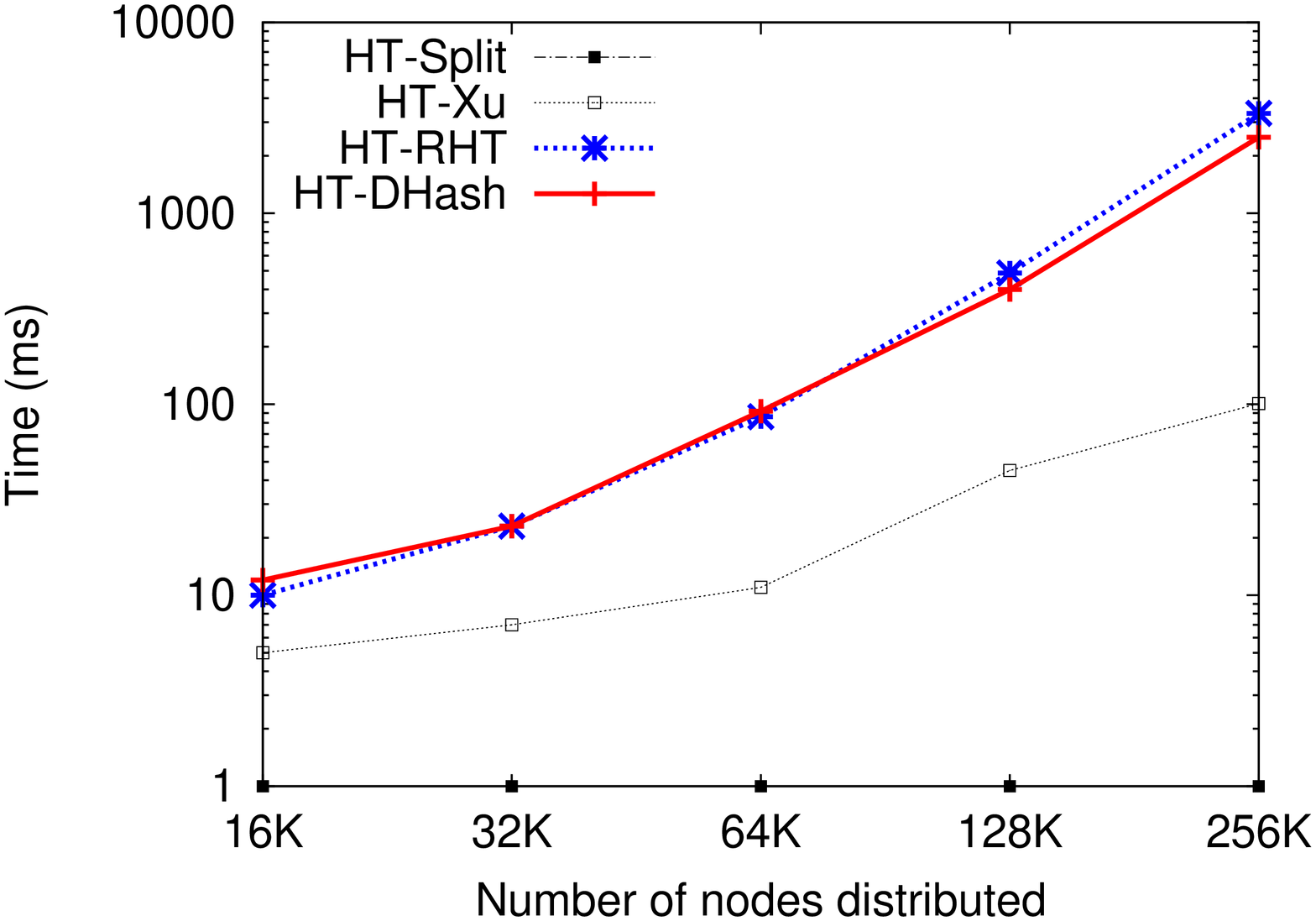}
        \label{fig.rebuild.LD.5.5.90}
    }
    \subfloat[$1/3$ Insert, $1/3$ Delete]{
        \includegraphics[scale=0.16, angle=-90]{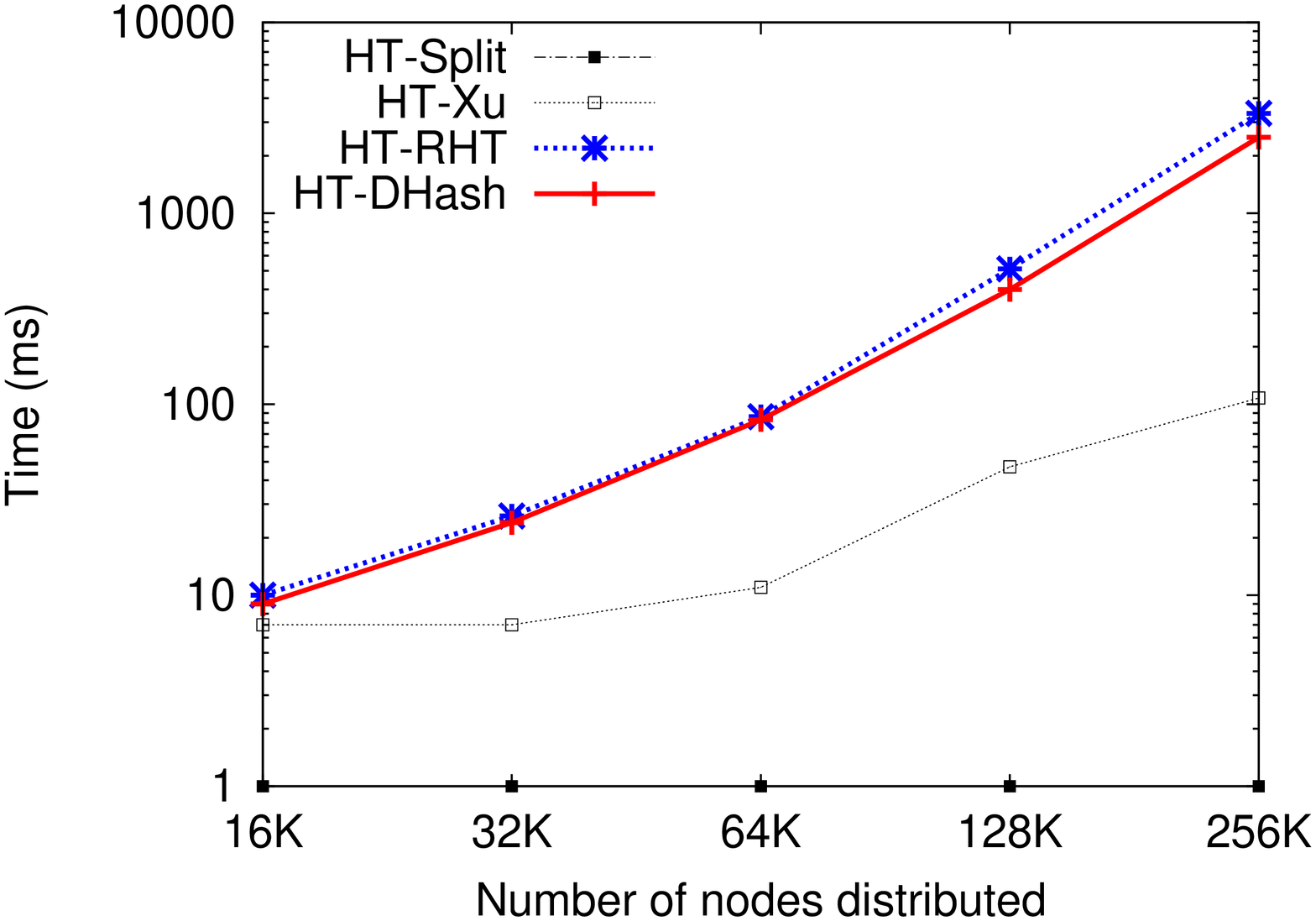}
        \label{fig.rebuild.LD.33.33.34}
    }
    \caption{Rebuilding efficiency.}
    \label{fig.eva.rebuild}
\end{figure}

Another observation is that operation mixes do not noticeably affect the
rebuilding efficiency of all evaluated hash tables,
as shown by the comparison of Figures \ref{fig.rebuild.LD.5.5.90} and \ref{fig.rebuild.LD.33.33.34}.
This observation suggests that for \sns,
given a hash table with a specified number of nodes,
programmers can predict how long the algorithm will take to rebuild the hash table.


\subsection{Performance on different architectures}\label{sec.eva.arch}

We now evaluate the overall performance of \sn on ARM and PowerPC,
other two important architectures in industry.
The characteristics of the servers used were listed in Table \ref{table.platforms}.
The benchmarking framework is the same as in Section \ref{sec.eva.perf}.
Experimental results with different average load factors
are marked with different suffixes in Figure \ref{fig.arch}.
For example, \emph{HT-DHash-20} shows the results of \sn with the average load factor of 20.
 
\begin{figure}[!h]
    \centering
    \subfloat[IBM Power9]{
        \includegraphics[scale=0.16, angle=-90]{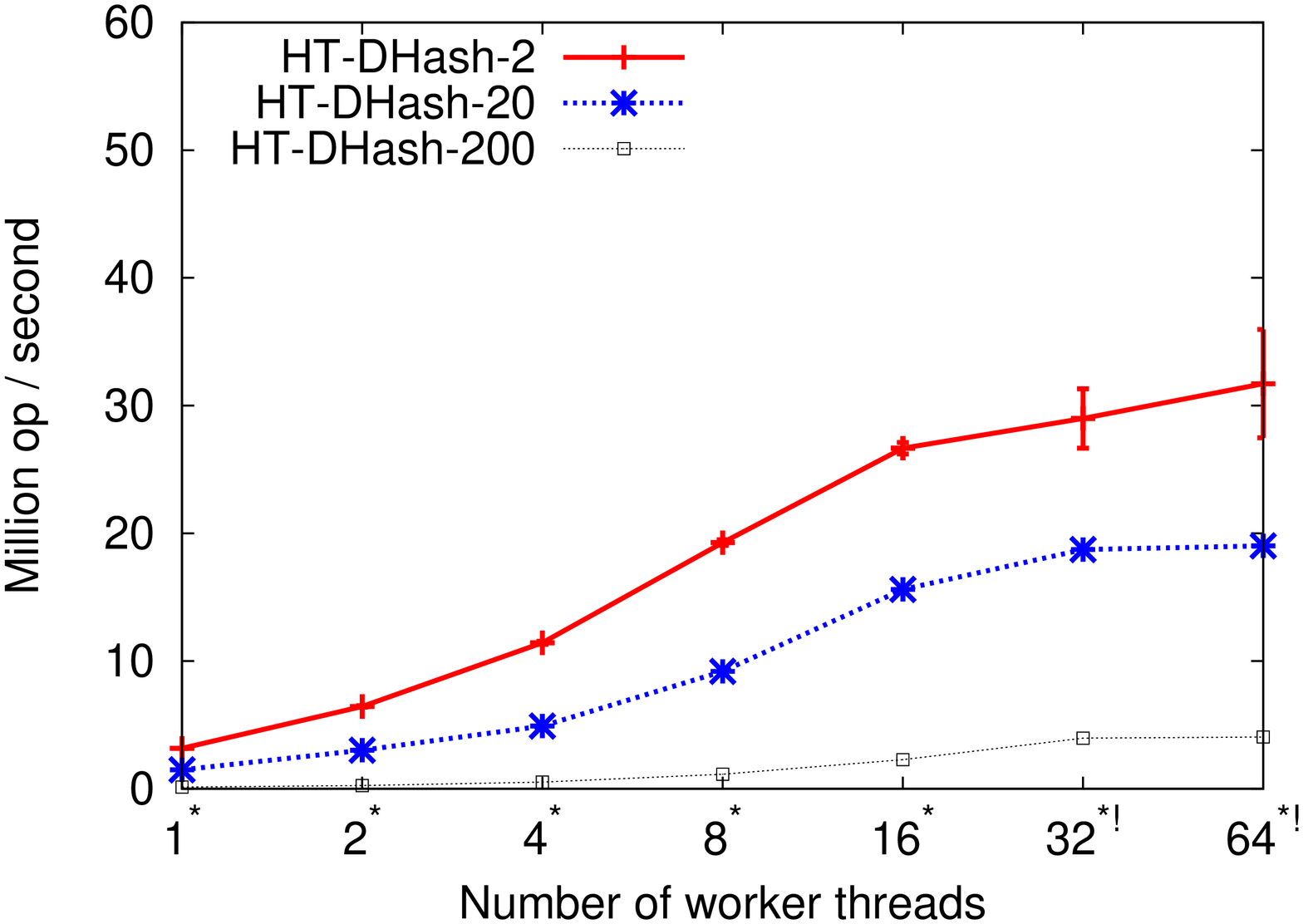}
        \label{fig.arch.PPC64}
    }
    \subfloat[Cavium ARMv8]{
        \includegraphics[scale=0.16, angle=-90]{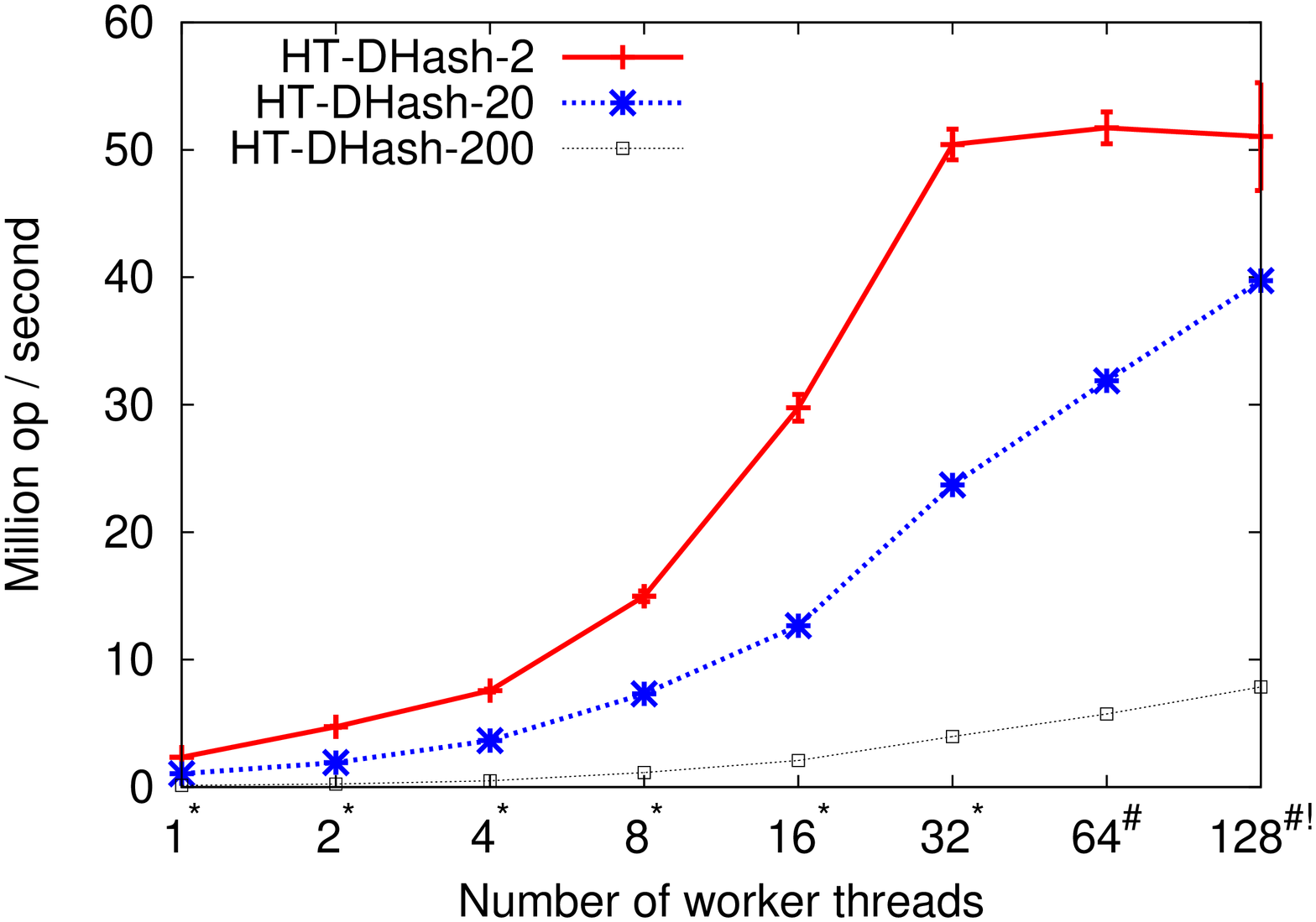}
        \label{fig.arch.ARM64}
    }
    \caption{Performance \sn on PowerPC and ARM.}
    \label{fig.arch}
\end{figure}

We make the following observation.
On both architectures, \sn scales well.
The overall performance of \sn increases nearly linearly
when the number of worker threads increases,
until worker threads oversubscribe CPU cores.
After that, the performance of \sn increases slightly or stays constant, 
but does not decrease.
Figure \ref{fig.arch} shows that
even if the average load factor of the hash table reaches 200,
\sn can provide the throughput of 4.1 and 7.9Mop/s
on IBM Power9 and ARMv8, respectively,
indicating that 
\sn is the algorithm of choice for real applications with heavy workloads.

\section{Conclusions}\label{sec.conclusion}

To overcome the hash collision problem, 
this paper presents \sns, a flexibly, efficient hash table
that can dynamically change its hash function on the fly.
\sn allows programmers to create specific implementations
that meet their requirements in terms of the algorithm's progress guarantee and performance.
We present the core technique to efficiently distribute nodes
from the old hash table to the new one in rebuilding,
and show that the result is highly scalable and robust using a variety of benchmarks
on three types of architectures.


\bibliographystyle{ACM-Reference-Format}
\bibliography{ms}


\end{document}